\documentclass[sigconf, nonacm]{acmart} 

\newcommand\vldbdoi{XX.XX/XXX.XX}
\newcommand\vldbpages{XXX-XXX}
\newcommand\vldbvolume{17}
\newcommand\vldbissue{1}
\newcommand\vldbyear{2024}
\newcommand\vldbauthors{\authors}
\newcommand\vldbtitle{\shorttitle}
\newcommand\vldbavailabilityurl{URL_TO_YOUR_ARTIFACTS}
\newcommand\vldbpagestyle{plain}

\usepackage{multirow}
\usepackage{enumitem}
\setenumerate[1]{itemsep=0pt,partopsep=0pt,parsep=\parskip,topsep=2pt, leftmargin=15pt}
\setitemize[1]{itemsep=0pt,partopsep=0pt,parsep=\parskip,topsep=2pt, leftmargin=10pt}
\usepackage[tight,footnotesize]{subfigure}
\usepackage[linesnumbered,ruled,vlined]{algorithm2e}

\SetCommentSty{mycommfont}
\SetInd{1.25ex}{1.25ex}
\DontPrintSemicolon
\usepackage{stfloats}
\usepackage{color}
\usepackage{url}
\usepackage{hyperref}
\usepackage{mathtools}
\usepackage{booktabs}
\usepackage[braket, qm]{qcircuit}
\usepackage{graphicx}
\usepackage{apxproof}

\begin{document}

\title{First Tree-like Quantum Data Structure: Quantum B+ Tree}

\author{Hao Liu}
\affiliation{%
  \institution{The Hong Kong University of Science and Technology}
}
\email{hliubs@cse.ust.hk}

\author{Xiaotian You}
\affiliation{%
  \institution{The Hong Kong University of Science and Technology}
}
\email{xyouaa@ust.hk}

\author{Raymond Chi-Wing Wong}
\affiliation{%
  \institution{The Hong Kong University of Science and Technology}
}
\email{raywong@cse.ust.hk}

\begin{abstract}
  Quantum computing is a popular topic in computer science, which has recently attracted many studies in various areas such as machine learning, network and cryptography. However, the topic of quantum data structures seems long neglected. There is an open problem in the database area: Can we make an improvement on existing data structures by quantum techniques? Consider a dataset of key-record pairs. Given an interval as a query range, a B+ tree can report all the records with keys within this interval, which is called a \emph{range query}. A classical B+ tree answers a range query in $O(\log N +k)$ time, where $N$ is the total number of records and the output size $k$ is the number of records in the interval. It is asymptotically optimal in a classical computer but not efficient enough in a quantum computer, because it is expected that the execution time and the output size are linear in a quantum computer. 

  In this paper, we propose the quantum range query problem. Different from the classical range queries, a quantum range query returns the range query results in \emph{quantum bits}, which has broad potential applications due to the foreseeable future advance of quantum computers and quantum algorithms. To the best of our knowledge, we design the first tree-like quantum data structure called the \emph{quantum B+ tree}. Based on this data structure, we propose a hybrid quantum-classical algorithm to do the range search. It answers a static quantum range query in $O(\log_B N)$ time, which is asymptotically optimal in quantum computers. Since the execution time does not depend on the output size (i.e., $k$, which could be as large as $O(N)$), it is significantly faster than the classical data structure. Moreover, we extend our quantum B+ tree to answer the dynamic and $d$-dimensional quantum range queries efficiently in $O(\log^2_B N)$ and $O(\log^d_B N)$ time, respectively. Our experimental results show that our proposed quantum data structures achieve up to 1000x improvement in the number of memory accesses compared to their classical competitors.

\end{abstract}

\maketitle

\setlength{\textfloatsep}{1pt}
\setlength{\intextsep}{1pt}
\setlength{\floatsep}{1pt}
\setlength{\abovecaptionskip}{1pt}
\setlength{\belowcaptionskip}{1pt}
\setlength{\abovedisplayskip}{1pt}
\setlength{\belowdisplayskip}{1pt}

\section{Introduction}\label{intro}
Consider a dataset where each item in this dataset is in the form of a key-record pair. The range query problem, which is to report all the items with keys within a given query range, has been a longstanding problem extensively studied and applied in broad applications. For example, a teacher may need to list all the students who obtained 40-60 marks in an exam, a smartphone buyer may need to list all the smartphones that sit around \$200-\$400, and a traveler may want to list all the nearby spots. 
{\color{black}
Apart from directly returning all the results to the user, range queries have also been applied as a subroutine of many problems in various areas. One typical example is the recommendation systems~\cite{gupta2013location,covington2016deep}, where range queries are commonly performed to extract a relatively small set of candidates first from the entire dataset as per user's request and then a recommendation algorithm is executed to find the top recommendations from the candidates.}

{\color{black}
Consider a movie dataset $D$ containing a list of movies each described as a key-record pair where the key is the release date and the record is a feature vector which represents some attributes such as name, genre and cast. Assume that Alice wants to find an interesting movie of 1990s.
We first perform a range query of finding a candidate set $D'$ of all the movies with release date from 1990 to 1999 and then give recommendations based on $D'$.}

{\color{black}
To answer range queries, many data structures \cite{bayer2002organization,comer1979ubiquitous,o1996log} have been proposed in classical computers to store the key-record pairs. One representative and widely-used data structure is the B+ tree~\cite{comer1979ubiquitous}.
It answers a range query in $O(\log N +k)$ time, where 
$N$ is the total number of key-record pairs and the output size $k$ is the number of records in the query range. Based on~\cite{yao1981should}, the B+ tree is asymptotically optimal for range queries in classical computers.
Note that for all range queries, the linear time of generating all the $k$ results is inevitable,
which is not efficient enough when range queries are used as subroutines of algorithms in quantum computers.}

Recently, quantum algorithms have attracted a lot of attention. Many quantum algorithms \cite{shor1994algorithms,grover1996fast,ruan2017quantum,adhikary2020supervised,li2021vsql,kerenidis2017quantum} have been proposed and are expected to show quadratic or even exponential speedup compared to classical algorithms, and thus many linear-complexity problems can be solved in sub-linear time in quantum computers. 
{\color{black}
For instance, a recent quantum recommendation system~\cite{kerenidis2017quantum} shows poly-logarithm time complexity to the input size.
However, when we consider the scenario where the recommendation is made among the $k$ candidates in a user-specified query range, 
the potential of quantum algorithms is limited since all the efficiency of the sublinear-time quantum algorithms is ruined by the linear time of generating all the $k$ candidates of traditional range queries.}

Motivated by the critical limitations of the classical range query in quantum computers,
{\color{black}we propose the problem called the \emph{quantum range query}, which returns the answer in \emph{quantum bits}, since we notice that quantum algorithms normally have the input in the form of quantum bits.}
Specifically, the quantum range query $QUERY(x, y)$ does not return a list of key-record pairs, but returns the quantum bits in a \emph{superposition} of all the key-record pairs with keys within the query range $[x, y]$, where the superposition is the ability of a quantum system to be in multiple states simultaneously. For example, consider a movie dataset $\{(key_0, rec_0),\cdots,$ $(key_{N-1}, rec_{N-1})\}$, where $N$ is the number of movies and each movie is represented as a key-record pair as mentioned above. Assume that we can construct a data structure in the quantum computer to store all the key-record pairs. If Alice wants to obtain all the movies of the 1990s, a quantum range query $QUERY(1990, 1999)$ will search on the data structure and return the quantum bits in a superposition of all the movies $(key_i, rec_i)$ whose $key_i$ is within the interval $[1990, 1999]$. Motivated by the different scenarios in real-world applications, we study the quantum range query problems in three different cases. The first case is that the dataset is immutable, so the problem is called the \emph{static quantum range query}. The second case is that insertions and deletions are supported in the dataset, and the problem is called the \emph{dynamic quantum range query}. The third case is that the keys in key-record pairs are multi-dimensional points (e.g., the locations of restaurants), so the problem is called the \emph{multi-dimensional quantum range query}.

{\color{black}
The quantum range queries have many potential applications.
(1) In relational database queries, it is common to have a numerical query range as a query condition (e.g., finding the movie with the highest ranking for release date between 1990 and 1999). 
To execute such query, a range query 
is commonly considered for query optimization. Specifically, a range query of retrieving all movies released between 1990 and 1999 is first performed to narrow down the search space to a small subset, and then finding the highest ranking movie could be executed efficiently on the subset. In the literature, some quantum algorithms have been proposed to support efficient database queries such as finding the database record that has the highest value on a certain field. When enabling query optimization, if we use a quantum range query to obtain the narrowed search space, the result of the quantum range query (in quantum bits) could be applied in those quantum algorithms for fast database queries.}
(2) In a quantum recommendation system as we mentioned previously,
a range query is first used as a filter to find a small set of candidates from the entire dataset,
and then a quantum recommendation algorithm~\cite{kerenidis2017quantum}
is used to make the top recommendation from the candidates.
A quantum range query, instead, returns a superposition of the candidates
which can be \emph{directly} applied to these \emph{existing} quantum recommendation algorithms.
(3) Data binning~\cite{dougherty1995supervised} is a popular approach
to improve the accuracy of various machine learning algorithms~\cite{berg2021deep, xue2017efficient},
which applies range queries to group similar items within a query range (i.e., a data bin) together.
Since it is foreseeable that similar quantum machine learning algorithms~\cite{adhikary2020supervised}
could also benefit from data binning,
the quantum range queries give the quantum binning results where the items in each bin are output
in the form of a superposition that can be \emph{readily} used in quantum machine learning algorithms.
{\color{black}
(4) Campos et.al. \cite{campos2010simple} studied how to use KNN to do movie recommendations. They found that it can improve the accuracy to only consider the rating datasets from the last month before the recommendation time in different years.  That means they need to select the records produced in the same month but in different years. Given a B+ tree containing all the movies sorted by the dates in the year, the calculation costs $O(\log N + N/12)$ time since the range query returns $N/12$ records on average. Assume that we can use a quantum range query to select the $N/12$ records in $O(\log N)$ time. Note that there is an existing $O(d^3)$ quantum KNN algorithm \cite{ruan2017quantum} where $d$ is the dimension of the feature vector. Therefore, the whole process only costs $O(\log N + d^3)$. In most cases, $d$ is much smaller than $N$. Assuming $d=O(\log N)$, we can find that the quantum time-periodic-biased KNN is exponentially faster than its classical competitor.}

Besides the above applications, we expect that it will be common for quantum algorithms
to have the input and output in the form of superpositions.
{\color{black}
For example, the HHL algorithm \cite{harrow2009quantum}
returns the answer to a linear system of equations in a superposition,
and it becomes a subroutine of the quantum SVM \cite{rebentrost2014quantum}.}
Therefore, we believe that in the future,
more quantum algorithms will benefit from the quantum range queries
which return the query results in superpositions.
In the database area, we expect the emergence of more quantum database search algorithms
(e.g., the quantum top-$k$ query), which could leverage the quantum data structures for query optimization
(e.g., applying the quantum range query to narrow down the searching space).

The quantum range query problem has two distinctive characteristics.
The first characteristic is that it allows the utilization of data structures,
which corresponds to building an index to accelerate the database queries
that has been a common approach in the database area.
Existing quantum algorithms~\cite{grover1996fast, boyer1998tight, grover2005partial, durr1996quantum}
do not consider a data structure and only focus on searching with unstructured data,
and thus they cannot solve the quantum range query problem efficiently.
Among them, the best quantum range query algorithm adapted from~\cite{boyer1998tight}
returns the result of $k$ items in $O(\sqrt{Nk})$ time,
which is inefficient when the dataset size $N$ grows large.

The second characteristic is the use of quantum computation to handle superpositions efficiently.
As mentioned previously, the $O(\log N + k)$ time complexity for a range query in classical computer
is already asymptotically optimal,
where the $O(k)$ cost of listing out all the $k$ results has no chance to be improved.
In the quantum range query problem, however, since we aim to obtain the $k$ results in the form of a superposition,
it is possible to eliminate this $O(k)$ cost and thus achieve a better time complexity
with the techniques in quantum computation.

Motivated by this, we propose the \emph{quantum B+ tree},
which is the first tree-like quantum data structure to the best of our knowledge.
Since the B+ tree~\cite{comer1979ubiquitous} is one of the most fundamental and widely-used data structures,
we believe that it is suitable to start a new world of quantum data structures with the B+ tree.
We design our quantum B+ tree with two components, the classical component and the quantum component.
The classical component follows the classical B+ tree,
which allows us to leverage its effective balanced tree structure.
The quantum component stores a concise ``replication'' of the hierarchical relationships in the B+ tree
in the quantum memory, which could load the relationships in the form of superpositions
efficiently in quantum computers due to quantum parallelism.\looseness=-1

Empowered by the two-component design of our quantum B+ tree,
we propose a hybrid quantum-classical algorithm called the \emph{Global-Classical Local-Quantum} (GCLQ) search
to solve the quantum range query problem.
It involves two main steps.
The first step is called the \emph{global classical search},
which finds a very small number (i.e., at most two) of candidate nodes from the classical B+ tree.
It is guaranteed that all the relevant results are covered in the candidate nodes
and account for a significant amount under the candidate nodes
(i.e., at least $\frac{1}{8B}$ of the items under candidate nodes are the relevant results).
Meanwhile, the global classical search is very efficient owing to the effective structure of B+ tree.
The second step is called the \emph{local quantum search},
which returns a superposition of all the exact query results from the candidate nodes
with efficient quantum parallelism techniques in the quantum memory.
As a result, the time complexity of our proposed GCLQ search is $O(\log_B N)$,
which is asymptotically optimal in quantum computers.
This improves the optimal classical result by reducing the $O(k)$ cost and
is much more efficient than the $O(\sqrt{Nk})$ time complexity of the existing quantum algorithms
without using any data structure.

We also propose two extensions of our quantum B+ tree
to solve the dynamic quantum range query and the multi-dimensional quantum range query.
Our two-component design with a classical B+ tree structure retained as a ``prototype''
allows us to flexibly extend our quantum B+ tree to the B+ tree variants in classical computers.
As such, we extend our quantum B+ tree to the \emph{dynamic quantum B+ tree}
by adapting the logarithmic method~\cite{bentley1980decomposable},
which supports inserting a new item into the tree and deleting an existing item from the tree.
All the insertions and deletions are also replicated to the quantum component of the tree,
which is efficient due to the conciseness of the quantum component.
We show that each insertion and deletion can be done in $O(\log_B N)$ time,
and the dynamic range query can be done in $O(\log^2_B N)$ with the dynamic quantum B+ tree.
To handle the multi-dimensional quantum range query,
we also propose the \emph{quantum range tree}
based on the classical range tree~\cite{bentley1978decomposable}
(which handles a classical $d$-dimensional range query in $O(\log^d_B N + k)$ time).
This is similar to the mechanism of constructing the quantum B+ tree from a classical B+ tree.
The multi-dimensional quantum range query can be answered in $O(\log^d_B N)$ time,
which also improves the classical $O(\log^d_B N + k)$ time complexity by the $O(k)$ cost.

{\color{black}
Therefore, we first propose the \emph{static quantum B+ tree}, which is the first tree-like quantum data structure to our best knowledge. We choose the B+ tree \cite{comer1979ubiquitous} as the first data structure to study in quantum computers. Since the B+ tree is the most fundamental and widely-used data structure, it is considered the most suitable one to open a new world of quantum data structures. To take the full advantages of quantum parallelism, we design a hybrid quantum-classical algorithm to answer the static quantum range query. The hybrid design is very popular, especially in quantum machine learning \cite{kim2001batch}. According to \cite{abohashima2020classification}, many studies such as \cite{adhikary2020supervised, chakraborty2020hybrid, benedetti2019generative,schuld2019quantum,schuld2020circuit, ruan2017quantum, mitarai2018quantum} used hybrid quantum-classical algorithms to do machine learning, since this design can reduce the circuit depth (which signifies the total number of instructions in the quantum algorithm) so that high performance can be obtained. Specifically, the same B+ tree is stored both in a quantum computer and a classical computer. The classical component of the tree-building algorithm will bulk load the data and construct the B+ tree in a manner like \cite{kim2001batch}. We formally define the concept of a quantum memory in Section \ref{sec:problem}, which is used to store mappings between bit-strings in a quantum computer. Each modification of the B+ tree in the classical computer will reflect in the quantum memory correspondingly. For example, if we add an edge from Node 2 to Node 3 in the B+ tree in the classical computer, we also add a mapping from 2 to 3 in the quantum memory. In the quantum computer, all the mappings from the nodes in the upper level to nodes in the lower level are maintained in the quantum memory. Our proposed quantum range query algorithm has the following two main steps.}

In summary, our contributions are shown as follows.
\begin{itemize}
    \item We are the first to study the quantum range query problems.
    \item We are the first to propose a tree-like quantum data structure, which is the quantum B+ tree.
    \item We design a hybrid quantum-classical
    algorithm that can answer a quantum range query on a quantum B+ tree in $O(\log_B N)$ time, which does not depend on the output size and is asymptotically optimal in quantum computers.
    \item We further extend the quantum B+ tree to the dynamic quantum B+ tree that supports insertions and deletions in $O(\log_B N)$ time, and the complexity of the quantum range query on the dynamic quantum B+ tree is $O(\log^2_B N)$.
    \item We also extend the quantum B+ tree to the quantum range tree, which answers a $d$-dimensional quantum range query in $O(\log^d_B N)$ time. 
    \item We conducted experiments to confirm the exponential speedup of the quantum range queries on real-world datasets. We considered both the time-based range query and the location-based range query, which are widely used in real-world applications. In our experiments, our quantum data structure is up to $1000\times$ faster than the classical data structure.
\end{itemize}

The rest of the paper is organized as follows.
In Section~\ref{sec:pre}, we first introduce some basic knowledge used in this paper about quantum algorithms.
We formally define the quantum range query problems in Section \ref{sec:problem}.
In Section~\ref{sec:alg}, we show the design of the quantum B+ tree and introduce our algorithm to answer the quantum range query.
In Section~\ref{sec:dyn_md}, we extend the quantum B+ tree to the dynamic and multi-dimensional versions.
Section \ref{exp} presents our experimental studies. 
In Sections~\ref{sec:related} and~\ref{sec:con}, we introduce the related work and conclude our paper, respectively.

\section{Preliminaries}\label{sec:pre}

A \emph{bit} is a basic unit when we store data in the memory or on disks.
In classical computers, a bit has two \emph{states} (i.e., 0 and 1).
A bit in quantum computers is known as a \emph{quantum bit},
which is called a \emph{qubit} in this paper for simplicity.
Similar to a bit in classical computers, a qubit also has states.
Following~\cite{dirac1939new, bayer2002organization},
we introduce the ``Dirac'' notation (i.e., ``$\ket \cdot$'') to describe the states of a qubit.
Specifically, we write an integer in $\ket \cdot$ (e.g., $\ket 0$) to describe a \emph{basis state},
and we write a symbol in $\ket \cdot$ (e.g., $\ket q$) to describe a \emph{mixed state}
(more formally known as a \emph{superposition}). 
Note that the symbol in a superposition of a qubit is used to denote this qubit.
That is, we say that the qubit $q$ has the superposition $\ket q$.\looseness=-1

We have two basis states of a qubit, $\ket 0$ and $\ket 1$,
which correspond to state 0 and state 1 of a classical bit, respectively.
Different from a classical bit, a qubit $q$ could have a state ``between'' $\ket 0$ and $\ket 1$,
which is described by a superposition $\ket q$.
Specifically, superposition $\ket q$ is represented as a linear combination of the two basis states:
$\ket q = \alpha \ket 0 + \beta \ket 1$,
where $\alpha$ and $\beta$ are two complex numbers called the \emph{amplitudes}
and we have $|\alpha|^2 + |\beta|^2 = 1$ 
($|\alpha|^2$ denotes the absolute square of a complex number $\alpha$).
In quantum computers, instead of directly obtaining the value of a bit,
we \emph{measure} a qubit $q$, which will ``collapse'' its superposition $\ket q$
and obtain the result state 0 with probability $|\alpha|^2$ or state 1 with probability $|\beta|^2$
(note that the sum of the probabilities is equal to 1).
For instance, let $q$ be a qubit with superposition $\ket q = (0.36+0.48i) \ket 0 + (0.48-0.64i) \ket 1$.
If we measure $q$, we can obtain 0 with probability $|0.36+0.48i|^2=0.36$ or 1 with probability $|0.48-0.64i|^2=0.64$.

Similar to a register in a classical computer,
we define a \emph{quantum register} to be a collection of qubits.
Consider a quantum register with two qubits.
The basis states of the two qubits will contain all the combinations of the basis states of single qubits,
i.e., $\ket{0} \ket{0}$, $\ket{0} \ket{1}$, $\ket{1} \ket{0}$ and $\ket{1} \ket{1}$.
If we measure one of the qubits, this measurement may change the state of the other qubit.
Such phenomenon is widely known as \emph{quantum entanglement} \cite{schrodinger1935discussion}.

Consider the following example.
For simple illustration, all the amplitudes in this example 
only have real parts (of complex numbers).
Let $\psi$ be a quantum register with two qubits $q_0$ and $q_1$,
where the measurement of $q_0$ will change the state of $q_1$.
Specifically, $\ket{q_1} = \frac{1}{\sqrt{2}}\ket 0 + \frac{1}{\sqrt{2}}\ket 1$,
which is measured to be 0 or 1 with probability both equal to $\frac{1}{2}$.
If 0 is obtained, the state of $q_0$ will be changed to $\ket{q_0}=0.6 \ket 0 + 0.8 \ket 1$,
and otherwise, it will be changed to $\ket{q_0}=\frac{1}{\sqrt{2}}\ket 0+\frac{1}{\sqrt{2}}\ket 1$.
As such, the state of $q_0q_1$ (i.e., $\psi$) can be represented as
\begin{align*}
\ket{\psi} \! = \! \ket{q_1} \ket{q_0} \! &= \! \frac{1}{\sqrt{2}}\ket 0(0.6 \ket 0 + 0.8 \ket 1)+\frac{1}{\sqrt{2}}\ket 1(\frac{1}{\sqrt{2}}\ket 0+\frac{1}{\sqrt{2}}\ket 1)\\
&= \! \frac{0.6}{\sqrt{2}}\ket{0}\ket{0} + \frac{0.8}{\sqrt{2}}\ket{0}\ket{1} + \frac{1}{2}\ket{1}\ket{0} + \frac{1}{2}\ket{1}\ket{1}.
\end{align*}
It indicates that the two \emph{entangled} qubits have 4 basis states $\ket{0} \ket{0}$, $\ket{0} \ket{1}$, $\ket{1} \ket{0}$ and $\ket{1} \ket{1}$
with amplitudes $\frac{0.6}{\sqrt{2}}$, $\frac{0.8}{\sqrt{2}}$, $\frac{1}{2}$ and $\frac{1}{2}$, respectively.
It can be verified that the sum of the absolute squares of all the 4 amplitudes is still equal to 1.
We can extend this system to a quantum register with $n$ qubits
which have $2^n$ basis states and the corresponding $2^n$ amplitudes.

The states of qubits can be transformed by \emph{quantum gates} in a \emph{quantum circuit},
which work similarly as the gates and circuits in classical computers.
We can also measure a qubit in a quantum circuit.
In the following, we introduce three common types of quantum gates that will be used in this paper,
namely a \emph{Hadamard gate}, an \emph{$X$-gate} and a \emph{controlled-$X$ gate}.
A \emph{Hadamard gate} \cite{hadamard1893resolution, nielsen2001quantum}
transforms $\ket 0$ into $\frac{1}{\sqrt{2}}\ket 0 + \frac{1}{\sqrt{2}}\ket 1$
and transforms $\ket 1$ into $\frac{1}{\sqrt{2}}\ket 0 - \frac{1}{\sqrt{2}}\ket 1$.
An \emph{$X$-gate}~\cite{nielsen2001quantum} swaps the amplitudes of $\ket 0$ and $\ket 1$.
A \emph{controlled-$X$} gate has two parts, the \emph{control} qubit, says $q_1$,
and the \emph{target} qubit, says $q_0$.
A controlled-$X$ gate applies an $X$-gate on the target qubit $q_0$ if the control qubit is $\ket 1$,
and otherwise, it keeps the target qubit unchanged.

For example, assume an input $\ket{q}=\frac{1}{\sqrt{2}}\ket{0}+\frac{1}{\sqrt{2}}\ket{1}$.
After being transformed by a Hadamard gate, we have
$\ket{q}=\frac{1}{\sqrt{2}}(\frac{1}{\sqrt{2}}\ket 0 + \frac{1}{\sqrt{2}}\ket 1)+\frac{1}{\sqrt{2}}(\frac{1}{\sqrt{2}}\ket 0 - \frac{1}{\sqrt{2}}\ket 1)=\ket{0}$.
Clearly, if we now measure $q$, we will always obtain 0.
We give more examples about the $X$-gate and the controlled-$X$ gate shortly.

A quantum circuit could involve the combination of multiple input qubits and multiple gates.
Since the state transformations and measurements for each qubit follow the timeline of a wire,
we could describe the state changes in an ordered list of events.
For example, we are given a quantum circuit with two input qubits $q_0$ and $q_1$ that is described as follows.
\begin{enumerate}
    \item Apply an $X$-gate on $q_1$.
    \item Apply a controlled-$X$ gate where the control qubit is $q_1$ and the target qubit is $q_0$.
\end{enumerate}
Consider the following three examples of input:
\begin{enumerate}
    \item $\ket{q_1}\ket{q_0}=\ket{0}\ket{0}$. After the $X$-gate on $q_1$, $\ket{q_1}\ket{q_0}=\ket{1}\ket{0}$.
    Since $\ket{q_1}=\ket{1}$, we apply an $X$-gate on $q_0$, and thus the result is $\ket{1}\ket{1}$.
    \item $\ket{q_1}\ket{q_0}=\ket{1}\ket{1}$. After the $X$-gate on $q_1$, $\ket{q_1}\ket{q_0}=\ket{0}\ket{1}$.
    Since $\ket{q_1}=\ket{0}$, we do nothing, and thus the result is $\ket{0}\ket{1}$.
    \item $\ket{q_1}\ket{q_0}=\frac{1}{\sqrt{2}}\ket{0}\ket{0}+\frac{1}{\sqrt{2}}\ket{1}\ket{1}$.
    After the $X$-gate on $q_1$, $\ket{q_1}\ket{q_0}\allowbreak=\frac{1}{\sqrt{2}}\ket{1}\ket{0}+\frac{1}{\sqrt{2}}\ket{0}\ket{1}$.
    The result is $\frac{1}{\sqrt{2}}\ket{1}\ket{1}+\frac{1}{\sqrt{2}}\ket{0}\ket{1}$.
\end{enumerate}

Following prior studies~\cite{grover1996fast, shor1994algorithms, wiebe2015quantum, zhang2018quantum},
we call such a quantum transformation consisting of a series of quantum gates in a quantum circuit as a \emph{quantum oracle}.
A quantum oracle can be regarded as a black box of a quantum circuit
where we only focus on the quantum operation it can perform but skip the detailed circuit.
Since different gate sets (which correspond to the instruction sets in classical computers)
may cause different time complexities and the general quantum computer is still at a very early stage,
we normally use the query complexity to study quantum algorithms,
which is to measure the number of queries to the quantum oracles.
In the quantum algorithm area, many studies~\cite{li2021sublinear, kapralov2020fast,
montanaro2017quantum, naya2020optimal, hosoyamada2018quantum, li2019sublinear, kieferova2021quantum}
assume that a quantum oracle costs $O(1)$ time, and analyze the time complexity based on this assumption.
We thus follow this common assumption in this paper.

\section{Problem Definition}\label{sec:problem}

In this section, we formally define the quantum range query problems
in the \emph{static}, \emph{dynamic} and \emph{multi-dimensional} cases, respectively.

We first consider the \emph{static} quantum range query problem
(for simplicity, it is simply called the quantum range query problem).
We are given an \emph{immutable} dataset $D$ of $N$ items.
Each item in $D$ is represented as a key-record pair $(key_i, rec_i)$ for $i \in [0, N - 1]$.
Note that the subscription $i$ in $(key_i, rec_i)$ represents the \emph{index} of this pair,
and the indices start from 0.
Each key $key_i$ is an integer, 
and each record $rec_i$ is a bit-string. 
Note that in our problem, the keys are integers,
but they can be easily extended to float numbers.
A range query of lower bound $x$ and upper bound $y$ is denoted as $QUERY(x, y)$.
In classical range query problem, $QUERY(x, y)$ returns a list of key-record pairs
whose keys fall in range $[x, y]$.
Let $k$ be the number of items that $QUERY(x, y)$ returns,
and let $l_0, l_1, \ldots, l_{k-1}$ be the indices of the returned items.
The returned list of $QUERY(x, y)$ is thus represented as
$\{(key_{l_0}, rec_{l_0}),\cdots, (key_{l_{k-1}}, rec_{l_{k-1}})\}$.
In the quantum range query problem,
we aim to return a superposition of all the key-record pairs in the list.
We formally define the quantum range query problem as follows.
\begin{definition}[Quantum Range Query]\label{defstatic}
    Given an immutable dataset $D=\{(key_0, rec_0), \ldots, (key_{N-1},$ $rec_{N-1})\}$
    and two integers $x$ and $y$ where $x \leq y$,
    a quantum range query $QUERY(x, y)$ is to return the following superposition
    $$\frac{1}{\sqrt{k}}\sum_{i=0}^{k-1}\ket{key_{l_i}}\ket{rec_{l_i}},$$
    such that for each $i \in [0, k-1]$, $x\leq key_{l_i} \leq y$.
\end{definition}

Note that in the quantum range query, we return a superposition
which is the linear combination of all the desired key-record pairs.
It can be observed that for each $i \in [0, k-1]$,
the probability of obtaining $\ket{key_{l_i}}\ket{rec_{l_i}}$ is equal to $\frac{1}{k}$.

Next, we also define the \emph{dynamic} quantum range query.
In the dynamic case, instead of having an immutable dataset,
we consider a \emph{dynamic} dataset $D$ that supports the \emph{insertion} and \emph{deletion} operations.
Specifically, an insertion operation inserts a new item $(key, rec)$ into $D$
and returns the new dataset $D \cup (key, rec)$.
A deletion operation delete an existing item $(key, rec)$ from $D$
and returns the new dataset $D \setminus (key, rec)$.
The dynamic quantum range query is formally defined as follows.

\begin{definition}[Dynamic Quantum Range Query]\label{defdynamic}
    Given a dataset $D=\{(key_0, rec_0), \ldots, (key_{N-1},$ $rec_{N-1})\}$
    that supports the insertion and deletion operations
    and two integers $x$ and $y$ where $x \leq y$,
    a dynamic quantum range query $QUERY(x, y)$ is to return the following superposition
    $$\frac{1}{\sqrt{k}}\sum_{i=0}^{k-1}\ket{key_{l_i}}\ket{rec_{l_i}},$$
    such that for each $i \in [0, k-1]$, $x\leq key_{l_i} \leq y$.
\end{definition}

Note that the only difference between the dynamic and the static quantum range query
is whether the given dataset is dynamic. 

Finally, we define the \emph{multi-dimensional} (static) quantum range query.
We are given a dataset $D$ of $N$ key-record pairs $(key_i, rec_i)$ (for $i \in [0, N-1]$),
where each key $key_i$ is in the form of a $d$-dimensional vector of integers,
i.e., $key_i = (key_{i, 1}, \ldots, key_{i, d})$.
A multi-dimensional quantum range query $QUERY(x, y)$
also takes lower bound $x$ and upper bound $y$ as input,
where $x$ and $y$ are also in the form of $d$-dimensional vectors of integers,
i.e., $x = (x_1, \ldots, x_d)$ and $y = (y_1, \ldots, y_d)$.
The multi-dimensional quantum range query is formalized as follows.

\begin{definition}[Multi-dimensional Quantum Range Query]\label{defmd}
    Given an immutable dataset $D=\{(key_0, rec_0), \ldots, (key_{N-1},$ $rec_{N-1})\}$
    where each $key_i$ is in the form of a $d$-dimensional vector
    (i.e., $key_i = (key_{i, 1}, \ldots, key_{i, d})$)
    and two $d$-dimensional vectors $x = (x_1, \ldots, x_d)$
    and $y = (y_1, \ldots, y_d)$ where $x_j \leq y_j$ for each $j \in [1, d]$,
    a multi-dimensional quantum range query is to return the following superposition
    $$\frac{1}{\sqrt{k}}\sum_{i=0}^{k-1}\ket{key_{l_i}}\ket{rec_{l_i}},$$
    where for each $i\in [0, k-1]$ and $j\in [1, d]$, $x_j\leq key_{l_i, j} \leq y_j$.
\end{definition}

\section{Quantum B+ Tree}\label{sec:alg}

In this section, we propose the quantum B+ tree to solve the quantum range query problems.
To the best of our knowledge, no prior studies explored using a data structure
in quantum computers.
Compared with performing a quantum algorithm on \emph{unstructured} dataset,
using quantum data structures could answer our quantum range query problems more efficiently.
We focus on the static quantum B+ tree in this section,
and discuss the dynamic and multi-dimensional variants later in Section~\ref{sec:dyn_md}.

Before introducing the details of the quantum B+ tree,
we first introduce the concept of the
\emph{\underline{Q}uantum \underline{R}andom \underline{A}ccess \underline{M}emory} (QRAM)
which is used to store the quantum data structure
and give query results in superpositions.
Like the classical RAM, we have the QRAM in quantum computers to read and write quantum states.
Following~\cite{kerenidis2017quantum, saeedi2019quantum},
we assume the existence of a \emph{classical-write quantum-read QRAM}.
{\color{black}It can store a classical value into a given address,
and it can accept a superposition of \emph{multiple} addresses
and returns a superposition of the corresponding values due to quantum parallelism.}
Formally, we define the QRAM as follows.

\begin{definition}[Quantum Random Access Memory (QRAM)]\label{def:QRAM}
A QRAM $\mathcal{Q}$ is an ideal model which {\color{black}can perform the following
store and load operations.}
\begin{itemize}
    \item A store operation denoted by $\mathcal{Q}[addr]=val$
    stores value $val$ into address $addr$,
    where $addr$ and $val$ are two bit-strings.
    \item A load operation denoted by
    \begin{equation}
        \mathcal{Q}\frac{1}{\sqrt{k}}\sum_{i=0}^{k-1}\ket{addr_i}\ket{a}=
        \frac{1}{\sqrt{k}}\sum_{i=0}^{k-1}\ket{addr_i}\ket{a \oplus val_i}
    \end{equation}
    takes a superposition of a number $k$ of addresses (i.e., $addr_i$ for $i \in [0, k-1]$)
    as input and loads a superposition of the values of these addresses
    into the destination quantum register of initial state $a$
    where $\oplus$ denotes the XOR operation.
\end{itemize}
\end{definition}

In our definition of QRAM, the concept of the store operation is similar to a classical RAM
where we store $val$ into a memory block at address $addr$.
For simplicity, we also use $\mathcal{Q}[addr]$ to denote the value stored at address $addr$.
In the load operation, the input of multiple addresses and the output of multiple values
are handled in superpositions, which is the major difference from a classical RAM.
The result of a load operation is written to a quantum register using the XOR operation
(i.e., $\oplus$).
For instance, if we set the initial state of the destination quantum register to be $\ket{0}$,
then the exact values will be loaded as follows.
\begin{equation}
    \mathcal{Q}\frac{1}{\sqrt{k}}\sum_{i=0}^{k-1}\ket{addr_i}\ket{0}=
    \frac{1}{\sqrt{k}}\sum_{i=0}^{k-1}\ket{addr_i}\ket{\mathcal{Q}[addr_i]}
\end{equation}
{\color{black}Although the QRAM is still discussed at the theoretical stage
and no existing physical implementation is yet available,
researchers believe that with the future advances of quantum computers,
the QRAM can be implemented to support the above store and load operation very efficiently
in the form of quantum oracles~\cite{li2021sublinear, kapralov2020fast,
montanaro2017quantum, naya2020optimal, hosoyamada2018quantum, li2019sublinear, kieferova2021quantum}.}

{\color{black}Note that even with the assumption of QRAMs, the quantum range query is still non-trivial.}
With the concept of QRAM, we can form the following method
of answering the quantum range query problems in two steps
using the existing quantum algorithms on the \emph{unstructured} dataset.
In the first step of pre-processing (for loading the dataset),
for each key-record pair $(key_i, rec_i)$,
we store it into a QRAM $\mathcal{Q}$ using the store operation $\mathcal{Q}[key_i]=rec_i$.
Clearly, the total time complexity is $O(N)$.
In the second step of data retrieval,
we first find all the addresses in $\mathcal{Q}$ that are within the query range,
which costs $O(\sqrt{Nk})$ time using the state-of-the-art
quantum algorithm~\cite{brassard1998quantum, boyer1998tight},
and then use the load operation to find the resulting superposition of the corresponding values.
Totally, the data retrieval step costs $O(\sqrt{Nk})$.
When $N$ grows large, this is even worse than the classical range query algorithms
(e.g., using the B+ tree) with time complexity logarithm to $N$.

To solve the quantum range query problems more efficiently with QRAM,
we propose our quantum data structure called the quantum B+ tree.
The major design of our quantum B+ tree is derived from
the basic idea of the classical B+ tree.
It is well known that the classical B+ tree utilizes a balanced tree structure,
which performs a range query in the dataset $D$ of size $N$ efficiently
with time complexity $O(\log N + k)$, where $k$ is the return size.
In classical computers, this time complexity is shown to be asymptotically optimal.
For the quantum range queries in quantum computers,
we are also interested in the best possible time complexity.
To explore that, we introduce another problem called the \emph{membership} problem,
which, given a query key, decides whether this key exists in the dataset $D$.
Existing studies have shown that the time complexity lower bound
of solving the membership problem in quantum computers is
$\Omega(\log N)$~\cite{ambainis1999better, sen2001lower}.
Based on that, we present the following lemma to show the time complexity lower bound
of solving our quantum range queries.
For the sake of space, we give a proof sketch for each lemma or theorem in this paper, and the full proof can be found in our technical report~\cite{technical_report}.

\begin{lemma}\label{lem:lower_bound_time_complexity}
The time complexity to answer a quantum range query is $\Omega(\log N)$. 
\end{lemma}
\begin{proofsketch}
Consider a quantum range query $QUERY(x, x)$.
It answers whether the key $x$ exists in the dataset,
which is a membership problem of time complexity $\Omega(\log N)$ by~\cite{sen2001lower}.
Therefore, the time complexity of answering a quantum range query is also $\Omega(\log N)$.
\end{proofsketch}

According to the above lemma, we cannot achieve a better time complexity than $O(\log N)$
for the quantum range query.
The goal of our static quantum B+ tree is thus to
obtain this asymptotically optimal time complexity $O(\log N)$.
Since the $O(k)$ time cost of listing all the $k$ returned results
cannot be reduced in classical range queries,
we consider eliminating this cost in our quantum B+ tree,
which is feasible in quantum computers that handle the superpositions
due to quantum parallelism.

Based on the above analysis, in the following,
we first introduce the details of how we design our quantum B+ tree in Section~\ref{subsec:design_quantum_b_plus},
and then, we propose our quantum range search algorithm called the \emph{global-classical local-quantum search} in Section~\ref{subsec:static_range_query_alg}.

\begin{figure*}[tbp]
  \centering
  \subfigure{
  \includegraphics[width=0.9\textwidth]{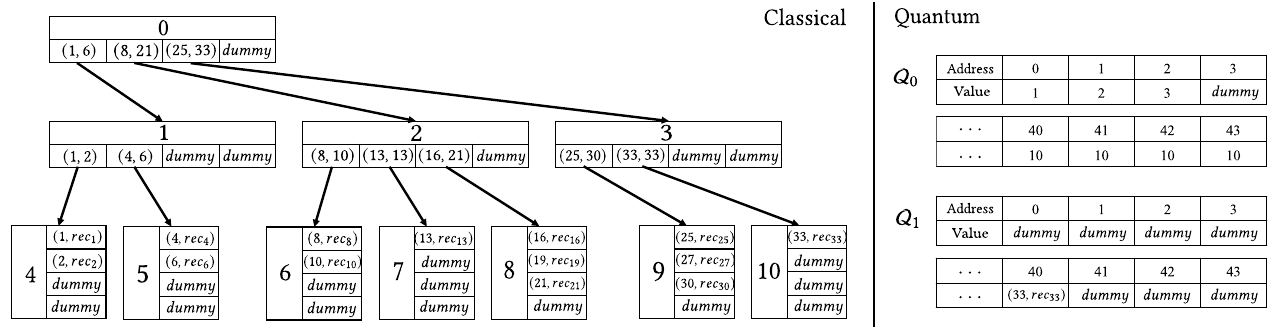}}
  \caption{An Example of a Quantum B+ Tree where $N=14$ and $B=4$}
  \label{fig:staticexp}
  \vspace*{-0.2cm}
\end{figure*}

\subsection{Design of Quantum B+ Tree}\label{subsec:design_quantum_b_plus}
{\color{black}
The design of our quantum B+ tree follows the core idea which focuses on representing the structure of the classical B+ tree in quantum computers. Specifically, the hierarchical relationships and node data are stored into the quantum memory (i.e., QRAM). This is to take advantage of the balanced tree structure of the classical B+ tree,
which has a tree \emph{height} of $O(\log N)$ and thus retains the $O(\log N)$ search efficiency. Also, it can utilize the QRAM to process the tree structure and return the range search results in superpositions efficiently by quantum parallelism.}

In the following, we first introduce some basic concepts of the classical B+ tree,
{\color{black}which is the prototype of our quantum B+ tree.}
We consider a \emph{weight-balanced B+ tree}~\cite{arge1996optimal}.
It is a tree data structure that consists of nodes in two types,
the internal nodes and the leaves.
Each node has a unique ID in the form of an integer, which is used to represent this node.
A user parameter $B$ called the branching factor is set to be
an integer that is a power of 2 and is at least 4.
A leaf contains $B$ key-record pairs sorted by the keys in ascending order.
An internal node has $B$ children, each of which is a leaf or an internal node.
Each node is also associated with a routing key, says $(L, U)$,
which represents a range with lower bound $L$ and upper bound $U$,
and all the key-record pairs under the internal node with routing key $(L, U)$
have keys within this range.
The children under an internal node have non-overlapping routing keys
and are sorted by the lower bound of their routing keys.
A node or a key-record pair under a leaf could be empty,
which is represented as a special mark called $dummy$.
We assume that $dummy$ appears only in the last positions under a node.

The left part in Figure~\ref{fig:staticexp} illustrates an example of
the classical component of a classical B+ tree 
with $N = 14$ key-record pairs and the branching factor $B = 4$.
It has four internal nodes of ID from 0 to 3 and seven leaves of ID from 4 to 10.
For instance, the node of ID 1 (simply called node 1) has routing key (1, 6),
indicating that all the data items with keys in range $[1, 6]$
are under node 1.
Also, node 1 has four children, two of which are non-dummy leaves
and the other two are dummy nodes.
For the leave shown as node 4 with routing key (1, 2),
two data items (with keys equal to 1 and 2) and two dummy items are assigned.

We define the \emph{level} of a node to be its distance to the root,
where the distance between two nodes is defined to be
the minimum number of edges connecting these two nodes.
In a B+ tree, all the leaves have the same level.
We also define the \emph{height} of a node to be its distance to a leaf,
and the height of the B+ tree is defined to be the height of the root.
We define the \emph{weight} of a node to be the number of
all non-dummy key-record pairs under this node.
A non-root node of height $h$ is said to be \emph{balanced}
if its weight is between $\frac{1}{4}B^{h+1}$ and $B^{h+1}$.
Further, a non-root node of height $h$ is said to be \emph{perfectly balanced}
if its weight is between $\frac{1}{2}B^{h+1}$ and $B^{h+1}$.\looseness=-1

In our running example shown in the left part of Figure~\ref{fig:staticexp},
there are 4 non-dummy key-record pairs under node 1 of height $h$ equal to 1,
and thus the weight of node 1 is equal to 4.
This indicates that node 1 is balanced (since 4 is no less than $\frac{1}{4}B^{h+1} = 4$)
but is not perfectly balanced (since 4 is less than $\frac{1}{2}B^{h+1} = 8$).

A B+ tree is said to be a weight-balanced B+ tree
if all its non-root nodes are balanced and its root has at least two non-dummy children.
It is easy to verify that the height of a weight-balanced B+ tree
with $N$ key-records pairs is $O(\log_B N)$.
Moreover, a weight-balanced B+ tree is said to be perfectly balanced
if all its non-dummy nodes are perfectly balanced.
It can be verified that the B+ tree shown in the left part of Figure~\ref{fig:staticexp}
is a weight-balanced B+ tree but is not perfectly balanced.

{\color{black}
After introducing the structure of a classical B+ tree, we now discuss how we represent the hierarchical relationships and node data of a B+ tree in QRAM. Since each node contains $B$ elements (including dummy), each of which has a child node with a routine key or a key-record pair, we could fit a classical B+ tree into a QRAM storing $M * B$ key-value pairs, where $M$ denotes the total number of nodes in the classical B+ tree, a key represents the ID of a tree node and a value stores the related information of this node. This can be done by a traversal of all the tree nodes in any order (in this paper, we choose the breadth-first order for simplicity).

We use two QRAMs, namely the hierarchy QRAM $\mathcal{Q}_0$ and the data QRAM $\mathcal{Q}_1$, to store the hierarchical relationships (i.e., the mapping from each node to its children) and the data (i.e., the routine key of a node or the key-record pair of a leave), respectively.}
Specifically, we perform $M * B$ store operations on $\mathcal{Q}_0$
where $M$ denotes the total number of nodes in the classical B+ tree as follows.
Let $C_{i, j}$ (for $i \in [0, M-1]$ and $j \in [0, B-1]$)
denote the node ID of the $j$-th child of node $i$ in the B+ tree.
Specially, if node $i$ is a leave with no children, we set $C_{i, j} = i$.
Also, if the $j$-th child of node $i$ is $dummy$, we assign $dummy$ to $C_{i, j}$.
Then, for each node $i \in [0, M-1]$ and each $j \in [0, B-1]$,
we perform a store operation to store value $C_{i, j}$ at address $i * B + j$,
i.e., $\mathcal{Q}_0[i * B + j] = C_{i, j}$.

In the right part of Figure~\ref{fig:staticexp},
we show an example of the quantum B+ tree. 
In the hierarchy QRAM $\mathcal{Q}_0$,
totally 44 ($= 11 * 4$) addresses are ``occupied''.
Among them, for instance, node 0 involves storing the IDs of its three non-dummy children
(i.e., 1--3) and one $dummy$ into addresses 0--3, respectively,
and node 10 involves storing value 10 for all its corresponding addresses 40--43
since it is a leaf.

{\color{black}
The data QRAM $\mathcal{Q}_1$ stores the routine key of each node or the key-record pair of each leave. 
Similar to $\mathcal{Q}_0$, we also perform $M * B$ store operations on $\mathcal{Q}_1$. Consider node $i$ in the B+ tree. If $i$ is an internal node, then for its $j$-th child, we assign the routine key (i.e., a 2-tuple of lower bound and upper bound of this routine key) to a variable, says $P_{i, j}$, which is $dummy$ if the $j$-th child is dummy. If $i$ is a leaf, we set $P_{i, j}$ to be the $j$-th key-record pair of node $i$ or $dummy$ if the $j$-th key-record pair is dummy.}
As such, for each node $i \in [0, M-1]$ and each $j \in [0, B-1]$,
we perform a store operation to store pair $P_{i, j}$ at address $i * B + j$,
i.e., $\mathcal{Q}_1[i * B + j] = P_{i, j}$.

Continuing our running example in Figure~\ref{fig:staticexp},
in the data QRAM $\mathcal{Q}_1$, we also store values into 44 addresses similarly,
where for node 0, since it is an internal node, all values stored are $dummy$ (at addresses 0--3),
and for node 10 with two key-record pairs and two dummies,
the values at address 40--43 are stored correspondingly.

{\color{black}Note that both QRAMs have data across the same number (i.e., $M * B$) of addresses, and at the same address, the value stored are related to the same node or the same key-record pair in a leaf. Thus, it is easy to combine the two QRAMs into one QRAM by using a larger quantum register storing the values of both QRAMs at the same address for implementation. In this paper, for clear structuring and illustration, we use two QRAMs to separate them.}

For both $\mathcal{Q}_0$ and $\mathcal{Q}_1$, we could easily retrieve multiple values
for the relationships in the tree with one load operation on QRAM.
Specifically, in the node QRAM $\mathcal{Q}_0$,
we perform the following load operation for node $i$ ($i \in [0, M-1]$)
\begin{equation}
\begin{split}
&\;\;\;\;\;\mathcal{Q}_0 \frac{1}{\sqrt{B}} \sum_{j=0}^{B-1}\ket{i*B+j}\ket{0}
= \frac{1}{\sqrt{B}}\sum_{j=0}^{B-1}\ket{i*B+j}\ket{C_{i, j}} \\
&= \frac{1}{\sqrt{B}}\sum_{j=0}^{f_i-1}\ket{i*B+j}\ket{C_{i, j}}+\frac{1}{\sqrt{B}}\sum_{j=f_i}^{B-1}\ket{i*B+j}\ket{dummy}
\end{split}
\end{equation}
where $f_i$ denotes the number of non-dummy children among all the children of node $i$.
This operation loads the IDs of all the children of node $i$ into a superposition
in $O(1)$ time.
Similarly, on $\mathcal{Q}_1$, we can load the key-record pairs under a leaf (with node ID $i$) with the following load operation
\begin{equation}
\begin{split}
&\;\;\;\;\;\mathcal{Q}_1 \frac{1}{\sqrt{B}}\sum_{j=0}^{B-1}\ket{i*B+j}\ket{0}
= \frac{1}{\sqrt{B}}\sum_{j=0}^{B-1}\ket{i*B+j}\ket{P_{i, j}} \\
&= \frac{1}{\sqrt{B}}\sum_{j=0}^{f_i-1}\ket{i*B+j}\ket{P_{i, j}}+\frac{1}{\sqrt{B}}\sum_{j=f_i}^{B-1}\ket{i*B+j}\ket{dummy} \text{.}
\end{split}
\end{equation}

\subsection{Global-Local Quantum Range Search}\label{subsec:static_range_query_alg}

{\color{black}In this section, we propose an algorithm called the \emph{global-local quantum range search} to solve the quantum range query with our quantum B+ tree.}

Recall that our quantum range query aims to return a superposition
of all the key-record pairs such that each key in the result is within a query range $[x, y]$.
To better present our algorithm, we first propose a technique called \emph{post-selection}
which can answer the quantum range query with our dataset stored in a QRAM even in an \emph{unstructured} way
(i.e., without using the quantum B+ tree).

The major idea of post-selection is based on the
efficient processing of superpositions of QRAM and the quantum oracle.
Specifically, we first apply the store operations to store all the key-record pairs 
in the dataset $D$ in a QRAM, says $\mathcal{Q}'$,
which is a \emph{pre-processing} step. 
That is, for each $i \in [0, N-1]$, we do a store operation $\mathcal{Q}'[i] = (key_i, rec_i)$.
We can thus obtain a superposition of all the key-record pairs as 
$\frac{1}{\sqrt{N}}\sum_{i=0}^{N-1}\ket{key_i}\ket{rec_i}$
by a load operation on $\mathcal{Q}'$ in $O(1)$ time.
After obtaining this superposition,
we design a quantum oracle as follows to \emph{post-select} the superposition
of the desired items within the query range $[x, y]$.

Let $in$ denote the set of indices of all the key-record pairs in the dataset
such that the keys are within $[x, y]$,
and let $out$ denote the set of all the remaining indices.
First, we place an additional qubit with initial state $\ket{0}$
in the end for the purpose of
controlling and measuring.
Thus, we obtain $\frac{1}{\sqrt{N}}\sum_{i=0}^{N-1}\ket{key_i}\ket{rec_i}\ket{0}$.
Next, we perform a quantum oracle which transforms the above superposition
based on a controlled-X gate controlled by whether $key_i$ is within $[x, y]$.
Specifically, if $x \leq key_i \leq y$, the last qubit will be transformed to $\ket{1}$,
and otherwise, no transformation is performed.
After the transformation, we obtain
$\frac{1}{\sqrt{N}}\sum_{i \in in}\ket{key_i}\ket{rec_i}\ket{1}
\allowbreak+\allowbreak\frac{1}{\sqrt{N}}\sum_{i \in out}\ket{key_i}\ket{rec_i}\ket{0}$.
Note that $k$ (resp. $N-k$) denotes the size of $in$ (resp. $out$).
The above state can be rewritten as
$\frac{\sqrt{k}}{\sqrt{N}} \cdot \frac{1}{\sqrt{k}} \sum_{i \in in}\ket{key_i}\ket{rec_i}\ket{1}
+\frac{\sqrt{N-k}}{\sqrt{N}} \cdot \frac{1}{\sqrt{N-k}} \sum_{i \in out}\ket{key_i}\ket{rec_i}\ket{0}$,
where the part $\frac{1}{\sqrt{k}} \sum_{i \in in}\ket{key_i}\ket{rec_i}$
is exactly the quantum range query result we desire.
Therefore, finally, we measure the last qubit,
which will obtain the desired superposition of all the key-record pairs in set $in$
(i.e., $\frac{1}{\sqrt{k}} \sum_{i \in in}\ket{key_i}\ket{rec_i}$)
with probability equal to $\frac{k}{N}$ (i.e., the last qubit is measured to be 1).
Since our desired result could not be deterministically obtained,
we repeat the final measurement step 
until the last qubit is measured to be 1.
It is easy to verify that $N/k$ repeats are needed in expectation.\looseness=-1

Overall, to answer the quantum range query,
the expected time complexity of the online processing of post-selection is $O(N/k)$,
since the final step is repeated $N/k$ times in expectation,
and all the other steps cost $O(1)$ based on our assumptions in quantum computers.
In range queries, the number of results $k$ is normally much smaller than the dataset size $N$,
and thus the $O(N/k)$ time cost could be too large.
Therefore, the post-selection technique as a standalone approach is inefficient,
since the dataset is organized in an unstructured manner in QRAM.

However, if $k$ and $N$ are approximately equal,
the cost of $O(N/k)$ could be small (e.g., close to an $O(1)$ cost) for a post-selection.
Motivated by this, in the following, we propose our efficient quantum range query algorithm
based on our quantum B+ tree, which leverages post-selection as a sub-step.

Our quantum range query algorithm is called
the \emph{\underline{G}lobal-\underline{C}lassical \underline{L}ocal-\underline{Q}uantum search} (GCLQ search).
It has two major steps.
The first step is the \emph{global classical} search,
which aims to find the candidate B+ tree nodes \emph{precisely}
from the classical B+ tree such that they contain all the query results
but the number of irrelevant results is as small as possible.
The second step is the \emph{local quantum} search,
which searches from the QRAMs in the quantum B+ tree
and returns the result in superposition efficiently
with a post-selection step.
In the following two subsection, we introduce the two steps, respectively.

\subsubsection{Global Classical Search}\label{subsubsec:global_classical}
We first introduce the global classical search step.
In this step, we traverse the classical B+ nodes from the root
to obtain the ``precise'' candidate nodes that are without too many irrelevant results.
Although the classical B+ is already capable of returning all the exact results,
we do not expect returning many detailed-level candidate nodes in this step
(which could cause the undesired cost of listing all the candidate nodes).
Instead, we postpone the processing of detailed-level nodes for fast local search
in the quantum component,
and we only return at most \emph{two} candidate nodes which are as ``precise'' as possible.

Specifically, given a B+ tree node with routing key $(L, U)$ and the query range $[x, y]$,
we classify this node into three categories,
namely an \emph{outside} node if range $[L, U]$ and range $[x, y]$ have no overlap,
a \emph{partial} node if $[L, U]$ is partially inside $[x, y]$,
and an \emph{inside} node if $[L, U]$ is completely inside $[x, y]$.
A partial node $i$ is said to be \emph{precise} if $i$ is a leaf or
there exists an inside node among the children of node $i$.
It can be verified that if a partial node is precise and the B+ tree is weight-balanced,
then a sufficiently significant portion of key-record pairs
under this node are included in the query range $[x, y]$
(which is formalized in a lemma shortly).

Now, we present the detail of our global classical search,
which aims to return precise partial nodes with level as small as possible
(to avoid processing detailed-level nodes).
For simplicity, we assume that the root node is a partial node
(otherwise it falls to special cases where the results involve no item or all items).
We create an empty list $\mathcal{L}$ of tree nodes
and add the root to the list initially.
Then, we loop the following steps until $\mathcal{L}$ is empty.
(1) We check whether there exists a precise node in $\mathcal{L}$.
(2) If the answer of Step (1) is yes, we immediately add all nodes
in $\mathcal{L}$ into the returned candidate set and terminate.
Otherwise, we replace each node in $\mathcal{L}$ with all its non-outside children
(since outside nodes cannot contain the desired results).

For example, consider a query $QUERY(5, 11)$ given to the quantum B+ tree in Figure~\ref{fig:staticexp}.
We first check the root (i.e., node 0) with three children.
We find that two of them (i.e., node 1 and 2) are partial nodes
and one of them (i.e., node 3) is an outside node.
Thus, the root is not precise (i.e., there does not exist a precise node in $\mathcal{L}$),
and we replace the root with node 1 and node 2 in $\mathcal{L}$.
Then, we check whether node 1 or node 2 is precise.
We find that node 2 is precise since it contains an inside child (i.e., node 6).
Therefore, we obtain the returned candidate set containing node 1 and node 2.

The following lemma shows the effectiveness 
of the global classical search.
\begin{lemma}\label{lem:global_classical}
    Given a query range $[x, y]$ and a B+ tree that is weight-balanced,
    the returned candidate set of the global classical search
    contains at most two nodes.
    Moreover, let $\mathcal{R}$ be the set of all key-record pairs
    under the above returned candidate nodes,
    and let $\mathcal{R}^*$ be the set of all key-record pairs
    such that the keys are within the query range $[x, y]$.
    Then, $\mathcal{R}^* \subset \mathcal{R}$ and
    $\vert \mathcal{R}^* \vert \geq \frac{1}{8B} \vert \mathcal{R} \vert$.
\end{lemma}
\begin{proofsketch}
    Since the routing keys of all nodes in the same level are disjoint,
    we cannot have more than two partial nodes in the same level.
    It is also easy to verify that the returned nodes are from the same level,
    Thus, the candidate set contains at most two nodes.
    Since we only filter out the outside node in this algorithm,
    we have $\mathcal{R}^* \subset \mathcal{R}$.
    Finally, since the precise partial node either is a leaf
    (which contains at least $1/B$ items in the query range),
    or contains an inside child
    (which contains at least $\frac{1}{4B}$ items in the query range due to the balanced nodes of the B+ tree),
    and it can be verify that at least one returned node is precise,
    we have $\vert \mathcal{R}^* \vert \geq \frac{1}{8B} \vert \mathcal{R} \vert$.
\end{proofsketch}

\subsubsection{Local Quantum Search}\label{subsubsec:local}
Now, we introduce the local quantum search.
Since the local quantum search starts from at most two candidate nodes, consider answering $QUERY(x, y)$ starting from a node $u$ and another node $v$ as the candidate nodes in level $j$ and the height of the tree is $h$. Step 1 is to initialize the first $n_i$ quantum qubits to be $\frac{1}{\sqrt{2}}(\ket{u}+\ket{v})$,
where $n_i$ denotes the number of bits to store each node index $i$.
Then, in Step 2, we add $n_b = \log_2 B$ auxiliary qubits $\ket{0}$ to the last and also apply a Hadamard gate on each auxiliary qubit. We obtain $$\frac{1}{\sqrt{2}}(\ket{u}+\ket{v})\mathcal{H}\ket{0}\mathcal{H}\ket{0}\cdots \mathcal{H}\ket{0}=\frac{1}{\sqrt{2B}}(\ket{u}+\ket{v})\sum_{i=0}^{B-1}\ket{i}.$$ This step is to enumerate all the edges from the candidate nodes.
Then, in Step 3, we add $n_i$ qubits to the end and apply $\mathcal{Q}_0$ to obtain all the children of the candidate nodes, so we obtain $$\frac{1}{\sqrt{2B}}\ket{u}\sum_{i=0}^{B-1}\ket{i}\ket{c_i}+\frac{1}{\sqrt{2B}}\ket{v}\sum_{i=0}^{B-1}\ket{i}\ket{c_{i+B}},$$ where $c_0, \cdots, c_{B-1}$ are the $B$ children of $u$ and $c_B, \cdots, c_{2B-1}$ are the $B$ children of $v$. If we only look at the last $n$ qubits, we obtain $\frac{1}{\sqrt{2B}}\sum_{i=0}^{2B-1}\ket{c_i},$ which is the $2B$ children of $u$ and $v$. We repeat Step 2 and Step 3 for $h-j$ times so that we obtain all the $2B^{h-j}$ leaves below $u$ and $v$. Then, we do the same thing as Step 2 to enumerate all the $2B^{h-j+1}$ key-record pairs in the $2B^{h-j}$ leaves. In the last step, we apply $\mathcal{Q}_1$ to obtain all the key-record pairs below $u$ and $v$ and then do a post-selection search. Denote the $2B^{h-j+1}$ key-record pairs as $(key_0, rec_0), \cdots, (key_{2B^{h-j+1}-1}, rec_{2B^{h-j+1}-1})$. Then the quantum state becomes $\frac{1}{\sqrt{2B^{h-j+1}}}\sum_{i=0}^{2B^{h-j+1}-1}\ket{key_i}\ket{rec_i}.$
We also use $\ket{in}$ to denote the $k$ key-record pairs in the query range and use $\ket{out}$ to denote the other dummy key-record pairs and non-dummy key-record pairs which are not in the query range. We obtain $\frac{\sqrt{k}}{\sqrt{2B^{h-j+1}}}\ket{in}+\frac{\sqrt{2B^{h-j+1}-k}}{\sqrt{2B^{h-j+1}}}\ket{out}.$ If we do a post-selection, we can obtain $\ket{in}$ with probability $\frac{k}{2B^{h-j+1}}$.

For example, consider a query $QUERY(5, 11)$ on the quantum B+ tree in Figure \ref{fig:staticexp}. As mentioned in Section \ref{subsubsec:global_classical}, the candidate nodes are node 1 and node 2. First, we initialize $\ket{\psi}=\frac{1}{\sqrt{2}}(\ket{1}+\ket{2})$.
After applying Hadamard gates and $\mathcal{Q}_0$, we obtain $\ket{\psi}=\frac{1}{\sqrt{8}}(\ket{4}+\ket{5}+\ket{6}+\ket{7}+\ket{8})+\frac{\sqrt{3}}{\sqrt{8}}\ket{dummy}$, which consists of all the children of node 1 and node 2. Then, after applying Hadamard gates and $\mathcal{Q}_1$, we obtain all the key-record pairs below node 1 and node 2, which is $\ket{\psi}=\frac{1}{\sqrt{32}}(\ket{1}\ket{rec_1}+\ket{2}\ket{rec_2}+\ket{4}\ket{rec_4}+\ket{6}\ket{rec_6}+\ket{8}\ket{rec_8}+\ket{10}\ket{rec_{10}}+\ket{13}\ket{rec_{13}}+\ket{16}\ket{rec_{16}}+\ket{19}\ket{rec_{19}}+\ket{21}\ket{rec_{21}})+\frac{\sqrt{22}}{\sqrt{32}}\ket{dummy}$. Finally, by a post-selection, we can obtain $\ket{\psi}=\frac{1}{\sqrt{3}}(\ket{6}\ket{rec_6}+\ket{8}\ket{rec_8}+\ket{10}\ket{rec_{10}})$ with probability $\frac{3}{32}$.\looseness=-1

Finally, we show the time complexity of our GCLQ algorithm in the following theorem.
\begin{theorem}\label{the:time_complexity_gclq}
    On average, the quantum range query algorithm returns the answer in $O(\log_B N)$ time.
\end{theorem}
\begin{proofsketch}
    We need to show that both the global classical search and the local quantum search costs $O(\log_B N)$ time.
    The former is because we return at most two candidates and the height of the tree is $O(\log_B N)$.
    The latter is because we need to repeat all the steps for at most $8B$ times, which is a constant time, and in each iteration, we do Step 2 and Step 3 for at most $O(\log_B N)$ times, so the local quantum search needs $O(\log_B N)$ time.
\end{proofsketch}
By Theorem \ref{the:time_complexity_gclq} and Lemma \ref{lem:lower_bound_time_complexity}, this algorithm is asymptotically optimal in a quantum computer.

\section{Dynamic and Multi-dimensional Variants}\label{sec:dyn_md}

In this section, we introduce two variant of our quantum B+ tree,
the dynamic quantum B+ tree (in Section~\ref{subsec:dynamic})
and the static quantum range tree (in Section~\ref{subsec:md}),
which solves the dynamic and the multi-dimensional range queries, respectively.

\subsection{Dynamic Quantum B+ Tree}\label{subsec:dynamic}
In this section, we introduce how to make the static quantum B+ tree dynamic.
One idea of designing our quantum B+ tree is to retain a classical B+ tree structure and
to maintain a concise replication of hierarchical relationships in QRAMs.
This enables the flexible extension to the B+ tree variants that have been well studied
in classical computers.
As such, we adapt the idea of the logarithmic method~\cite{bentley1980decomposable}
to enable the insertion operations of the quantum B+ tree
by building some forests on the quantum B+ tree.
Based on the forests, we propose our approach to perform deletions.
In the following, the details of insertion and deletion
are discussed in Section~\ref{subsubsec:ins_del},
and then in Section~\ref{subsubsec:dynamic_query},
we introduce how to solve the dynamic quantum range query with the dynamic quantum B+ tree.

\subsubsection{Insertion and Deletion}\label{subsubsec:ins_del}
Following the idea of the logarithmic method~\cite{bentley1980decomposable},
{\color{black}
we build at most $L$ forests, says $F_0, \cdots, F_{L - 1}$,
where $L = \lfloor\log_B N\rfloor+1$, and for each $i\in[0, L-1]$,
the forest $F_i$ contains at most $B-1$ static quantum B+ trees of height $i$.

To insert a new key-record pair $(key, rec)$, we insert it into a sorted list first.
When the length of this sorted list reaches $B$, we flush it into $F_0$.
Then, whenever a forest $F_i$ has $B$ quantum B+ trees,
which indicates that we have $B$ quantum B+ trees of height $i$,
we merge the $B$ quantum B+ trees of height $i$ into a quantum B+ trees of height $i+1$,
and add it into $F_{i+1}$.
}

In additional to the classical B+ tree, we also replicate the insertion in the quantum B+ tree
such that the hierarchical relationships in the QRAM are consistent with the classical component.
Note that the store operation costs $O(1)$ in the QRAM, which creates marginal extra cost for insertion.

\begin{figure*}[tbp]
  \centering
  \subfigure[Deletion of $(6, rec_6)$ in node 5: The weight of node 1 becomes $3$ which is less than $\frac{1}{4}B^2=4$, and thus node 1 becomes imbalanced. It borrows node 6 from node 2 such that node 1 and node 2 are both balanced.]{
    \label{fig:2-1}
    \includegraphics[width=0.54\textwidth]{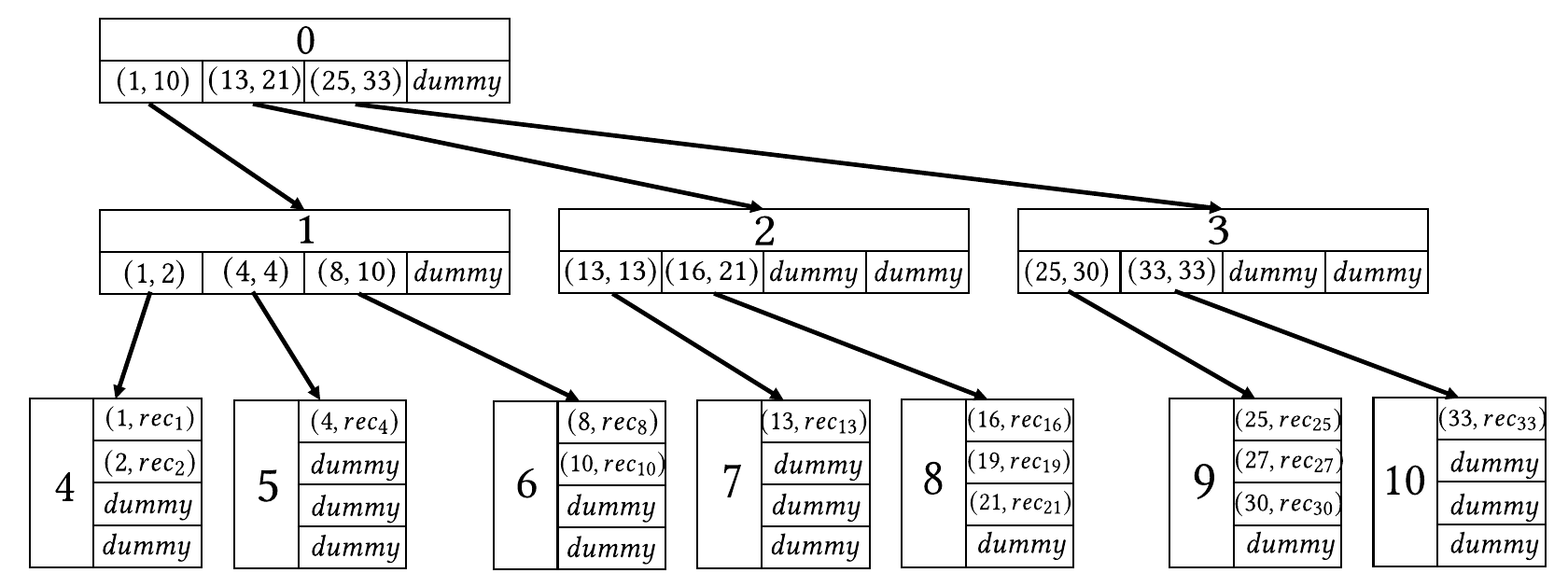}}
  \subfigure[Deletion of $(27, rec_{27})$ in node 9: Node 3 becomes imbalanced, and it cannot borrow node 8 from node 2. Node 2 and node 3 cannot be directly merged since they have 5 children, and thus the whole subtree is rebuilt.]{
    \label{fig:2-2}
    \includegraphics[width=0.44\textwidth]{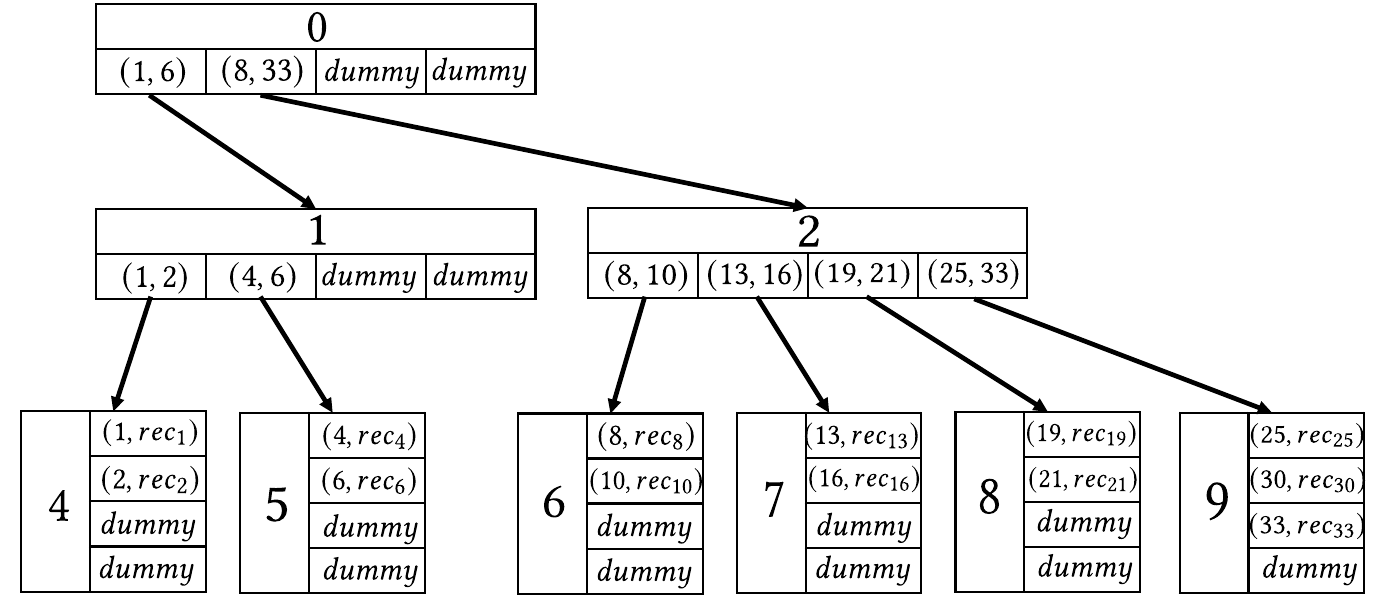}}
  \caption{Examples of Deletions}
  \label{fig:deleteexp}
  \vspace*{-0.2cm}
\end{figure*}

Next, we discuss how to delete a key-record pair in a dynamic quantum B+ tree.
It involves two steps.
The first step is to locate the quantum B+ tree in the forests
which contains the key-record pair.
The second step is to delete the key-record pair in both the classical component and
the quantum component of the quantum B+ tree.

To locate the tree containing a key-record pair,
we assign each key-record pair a unique ID in ascending order (e.g., by a counter).
We also maintain an auxiliary B+ tree $\mathcal{T}_0$
to store the one-to-one mapping from the key-record pair to its ID.
When a key-record pair is inserted, we insert the key-record pair and its ID into $\mathcal{T}_0$.
When a key-record pair is deleted, we delete it in $\mathcal{T}_0$ accordingly.
Furthermore, we maintain another auxiliary B+ tree $\mathcal{T}_1$
to store the mapping from the key-record pair ID to the forest $F_i$ it belongs to.
Similar to $\mathcal{T}_0$, we do insertions and deletions in $\mathcal{T}_1$ accordingly.
In addition, when we merge $F_i$, we update $\mathcal{T}_1$ for all the key-record pairs in $F_i$.

{\color{black}
To delete a key-record pair $(key, rec)$ in a B+ tree in $F_i$,
we do a classical search to find the leaf that contains $(key, rec)$ and
replace the key-record pair with a $dummy$.
Then, we check if its ancestors are still balanced.
If an imbalanced ancestor is found, we first check if it can borrow a child from its sibling.
Figure \ref{fig:2-1} shows an example in this case.
Otherwise, we check if it can be directly merged with its sibling,
which indicates that the node and its sibling have at most $B$ children.
If not, we merge the node and its sibling by rebuilding the subtrees below them.
Figure \ref{fig:2-2} shows an example in this case.
Then, after rebalancing the B+ tree, we check if the root node still has at least two children.
If not, we check if it can borrow a child from another B+ tree in $F_i$.
If not, then if there are at least two B+ trees in $F_i$,
we merge the root nodes of the two B+ trees.
Otherwise, we remove the root node and downgrade the B+ tree from $F_i$ to $F_{i-1}$.
}

The following lemma shows that the insertions and deletions in the quantum B+ tree is efficient.
\begin{lemma}\label{lem:insertion_cost}
    The amortized cost of an insertion or a deletion in a dynamic quantum B+ tree is $O(\log_B N)$.
\end{lemma}
\begin{proofsketch}
    By Theorem 3.1 in \cite{bentley1980decomposable}, $N$ insertions totally cost $O(N\log_B N)$ time. Therefore, the amortized cost of one insertion is $O(\log_B N)$.
    For deletions, the first step costs $O(\log_B N)$ time, since it consists of two point queries in B+ trees.
    By the analysis of partial rebuilding in \cite{overmars1983design}, the second step also costs $O(\log_B N)$ time.
    Therefore, the amortized cost of a deletion is $O(\log_B N)$.
\end{proofsketch}

\subsubsection{Query}\label{subsubsec:dynamic_query}
To answer a query $QUERY(x,y)$ on our dynamic quantum B+ tree,
we extend our idea of the GCLQ algorithm to perform both a global classical search
and a local quantum search on the forests $F_0, \cdots, F_{L-1}$.

First, for global classical search, we follow the similar global search steps
as introduced in Section~\ref{subsec:static_range_query_alg}
for each tree in the forests.
Specifically, we first create an empty list $\mathcal{L}_i$
for each forest $F_i$ ($i \in [0, L-1]$),
and we add all the root nodes in $F_i$ into $L_i$.
Then, for each list $\mathcal{L}_i$,
we perform a global classical search from each root in $\mathcal{L}_i$.
A candidate set of at most two nodes are returned for each root,
which are then added to the original list for easier processing.

Then, we perform the local quantum search for the candidate nodes in the lists
$\mathcal{L}_0,\cdots, \mathcal{L}_{L-1}$.
{\color{black}
Consider the nodes $u_0, u_1,\cdots, u_{m-1}$ in the lists, where $m$ is the total number of nodes in $\mathcal{L}_0,\cdots, \mathcal{L}_i$. We initialize the quantum bits to be $\sum_{i=0}^{m-1}\frac{\sqrt{B^{h(u_i)+1}}}{\sqrt{\sum_{j=0}^{m-1} B^{h(u_j)+1}}}\ket{u_i},$ where $h(u_i)$ is the height of $u_i$, and $\frac{\sqrt{B^{h(u_i)+1}}}{\sqrt{\sum_{j=0}^{m-1} B^{h(u_j)+1}}}$ is the normalized amplitude of each $u_i$ such that each of the key-record pair below the nodes in the lists has the same amplitude in the result.
Then, we do the same steps 
in Section \ref{subsubsec:local} to obtain the leaves below the nodes. Finally, we do a post-selection search as discussed in Section \ref{subsec:static_range_query_alg}.
}

The following theorem shows the time complexity for the dynamic quantum B+ tree
to answer a quantum range query.
\begin{theorem}\label{the:time_complexity_dynamic}
    On average, the dynamic quantum B+ tree answers a range query in $O(\log^2_B N)$ time.
\end{theorem}
\begin{proofsketch}
    The global classical search costs $O(\log^2_B N)$ time,
    since we have $O(\log_B N)$ forests and the global search for each forest needs $O(\log_B N)$ time.
    For the local quantum search, the initialization and the steps to obtain the leaves cost $O(\log_B N)$.
    The average number of post-selection is $O(\log_B N)$, since there are at most $2B\lfloor\log_B N\rfloor$ nodes in the lists.
    Totally, the time complexity is $O(\log^2_B N)$.
\end{proofsketch}

\subsection{Quantum Range Tree}\label{subsec:md}

In this section, we consider the multi-dimensional quantum range query.
The range tree~\cite{bentley1978decomposable} is a common data structure
for the classical multi-dimensional range query.
It answers a $d$-dimensional range query in $O(\log^d N+k)$ time.
For the multi-dimensional quantum range query,
we extend our mechanism of the static quantum B+ tree
to form the static quantum range tree.
Similarly, our static quantum range tree improves the complexity to $O(\log^d N)$,
which saves the $O(k)$ part compared with the classical range tree.

We construct a quantum range tree by induction.
{\color{black}
To build a $d$-dimensional quantum range tree,
we first build a classical B+ tree indexing the $d$-th dimension of the keys.
Then, we build a $(d-1)$-dimensional tree for each internal node in the B+ tree
based on the key-record pairs below the internal node.
Specially, a $1$-dimensional quantum range tree is a static quantum B+ tree.
Obviously, the space complexity and the time complexity of the construction
are similar to the classical range tree.

Then, we discuss the $d$-dimensional quantum range query. Starting from the $d$-th dimension, we search the B+ tree and find the $O(\log_B N)$ internal nodes in the search path which covers the $d$-th dimension of the query range. Then, we turn to the quantum range trees on the $O(\log_B N)$ internal nodes and do a $(d-1)$-dimensional range search. Similar to the classical range tree, we then fetch the answers in $O(\log_B^{d-1} N)$ $1$-dimensional quantum range trees.
}

In the following theorem, we also show the time complexity of
answering the $d$-dimensional quantum range query with our proposed quantum range tree.
\begin{theorem}\label{the:time_complexity_md}
    On average, the quantum range tree answers a $d$-dimensional quantum range query in $O(\log^d_B N)$ time.
\end{theorem}
\begin{proofsketch}
The first step to obtain the $O(\log_B^{d-1} N)$ 1D quantum range trees costs $O(\log_B^{d-1} N)$ time. The second step to perform a global classical search on the $O(\log_B^{d-1} N)$ quantum B+ tree costs $O(\log_B^d N)$ time. The third step to perform a local quantum search starting from $O(\log_B^{d-1} N)$ candidate nodes costs $O(\log_B^{d-1} N)$ time.
Overall, the average time complexity is $O(\log^d_B N)$.
\end{proofsketch}

\section{Experiment}\label{exp}

In this section, we show our experimental results to verify the superiority of our proposed quantum B+ tree
on quantum range queries.
The study of the real-world quantum supremacy \cite{arute2019quantum, boixo2018characterizing, terhal2018quantum},
which is to confirm that a quantum computer can do tasks faster than classical computers,
is still a developing topic in the quantum area.
Thus, in this paper, we would not verify the quantum supremacy,
which we believe will be fully verified in a future quantum computer.
To the best of our knowledge, there exists no information about the real implementation of the quantum memory
(i.e., QRAM) that can suggest the execution time of algorithms on quantum computers.
Therefore, we chose to evaluate the \emph{number of memory accesses} to make the comparisons,
which corresponds to IOs in traditional searches.
In the quantum data structures, a QRAM read or write operation is counted as 1 IO
(following the assumption of the $O(1)$ load/store operation in QRAM~\cite{kerenidis2017quantum, saeedi2019quantum}).
In the classical data structures, a page access is counted as 1 IO as well.
Moreover, although there exist some quantum simulators such as Qiskit \cite{wille2019ibm} and Cirq \cite{cirq},
they are not capable of simulating an efficient QRAM, 
and thus we chose to use C++ to perform the quantum simulations.

We used two real datasets from SNAP \cite{leskovec2016snap}, namely \emph{Brightkite} and \emph{Gowalla}.
Each of the two datasets contains a set of check-in records,
each of which consists of a timestamp (i.e., an integer)
and a location (i.e., a 2-tuple of integers).
The original sizes of the two datasets are 4M and 6M, respectively.
For the 1-dimensional static and dynamic quantum range queries,
we set the timestamp as the search key.
For the multi-dimensional quantum range query,
we set the location as the 2-dimensional search key.

For the 1-dimensional static and dynamic quantum range queries,
we compare our GCLQ search algorithm on our proposed quantum B+ tree and its dynamic variant
(which are simply denoted as \textbf{quantum B+ tree})
with the classical B+ tree search algorithm and its dynamic version~\cite{bentley1980decomposable}
(which are denoted as \textbf{classical B+ tree}).
For the multi-dimensional quantum range queries,
we also compare our search algorithm on the proposed quantum range tree
(which is denoted as \textbf{quantum range tree})
with the search on the classical range tree~\cite{chazelle1990lower1}
(which is denoted as \textbf{classical range tree}).
Note that the classical range tree is shown to be asymptotically optimal
to solve the multi-dimensional range queries on classical computers.

The factors we studied are the dataset size $N$, the branching factor $B$ and the selectivity
(which is defined to be the proportion of items within the query range among the entire dataset).
We varied $N$ from 4K to 2M (by randomly choosing a subset from each dataset with the target size).
We varied $B$ from 4 to 64 for each data structure.
We varied the selectivity from 1\% to 10\%.
By default, $N =$ 2M, $B =$ 16 and the selectivity is 5\%.
For each type of range query, we randomly generate 10000 range queries and report the average measurement.

We present our experimental results as follows.
For the sake of space, we mainly show the results of dataset \emph{Brightkite},
while we observe similar results in the other dataset \emph{Gowalla} in all our experiments.
The complete results of dataset \emph{Gowalla} can be found in our technical report~\cite{technical_report}.

\begin{figure*}[tbp]
  \centering
  \hspace*{-1.5cm}
  \includegraphics[width=0.95\linewidth]{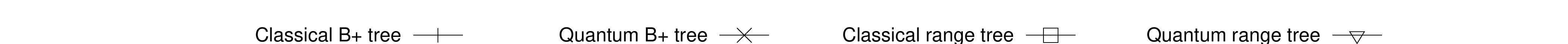}
  
  \begin{tabular}{c c c c c}
    \begin{minipage}[htbp]{0.21\linewidth}
      \hspace{-0.4cm}
      \includegraphics[width=\linewidth]{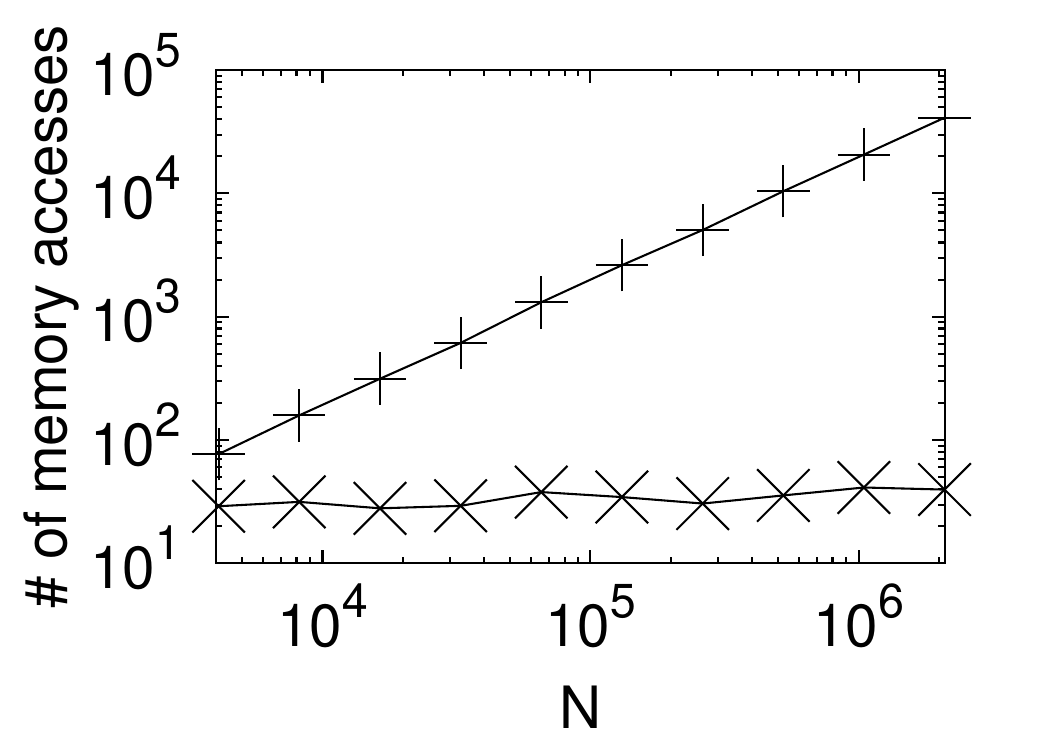}
    \end{minipage}
    &
    \begin{minipage}[htbp]{0.21\linewidth}
      \hspace{-0.9cm}
      \includegraphics[width=\linewidth]{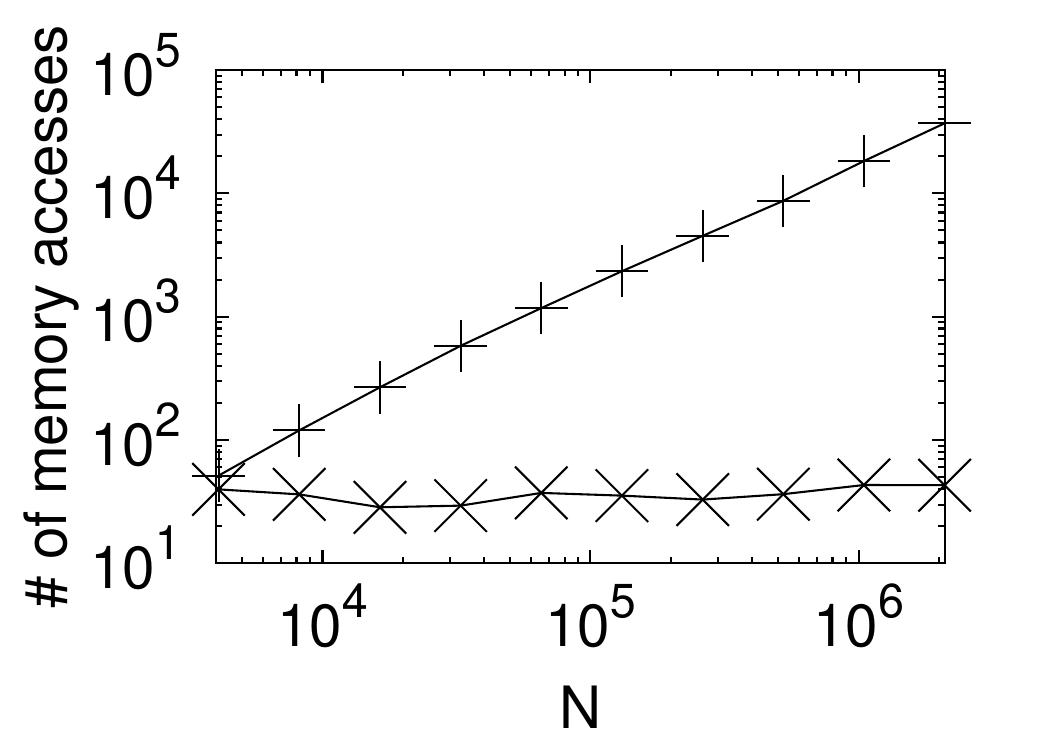}
    \end{minipage}
    &
    \begin{minipage}[htbp]{0.21\linewidth}
      \hspace{-1.4cm}
      \includegraphics[width=\linewidth]{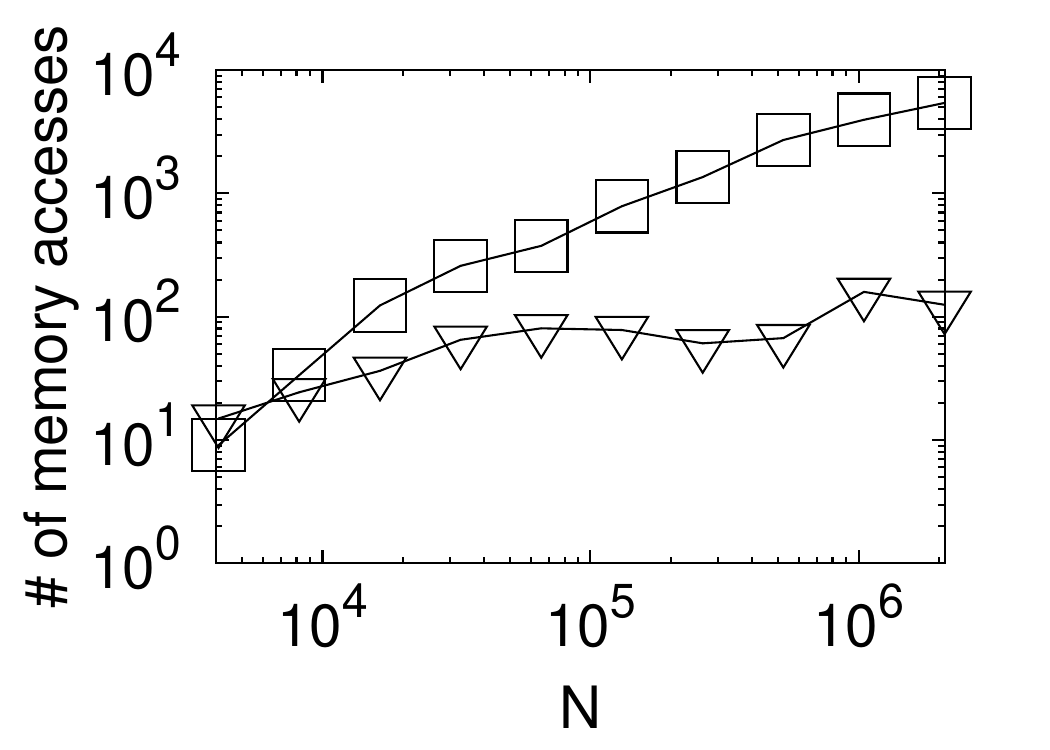}
    \end{minipage}
    &
    \begin{minipage}[htbp]{0.21\linewidth}
      \hspace{-1.85cm}
      \includegraphics[width=\linewidth]{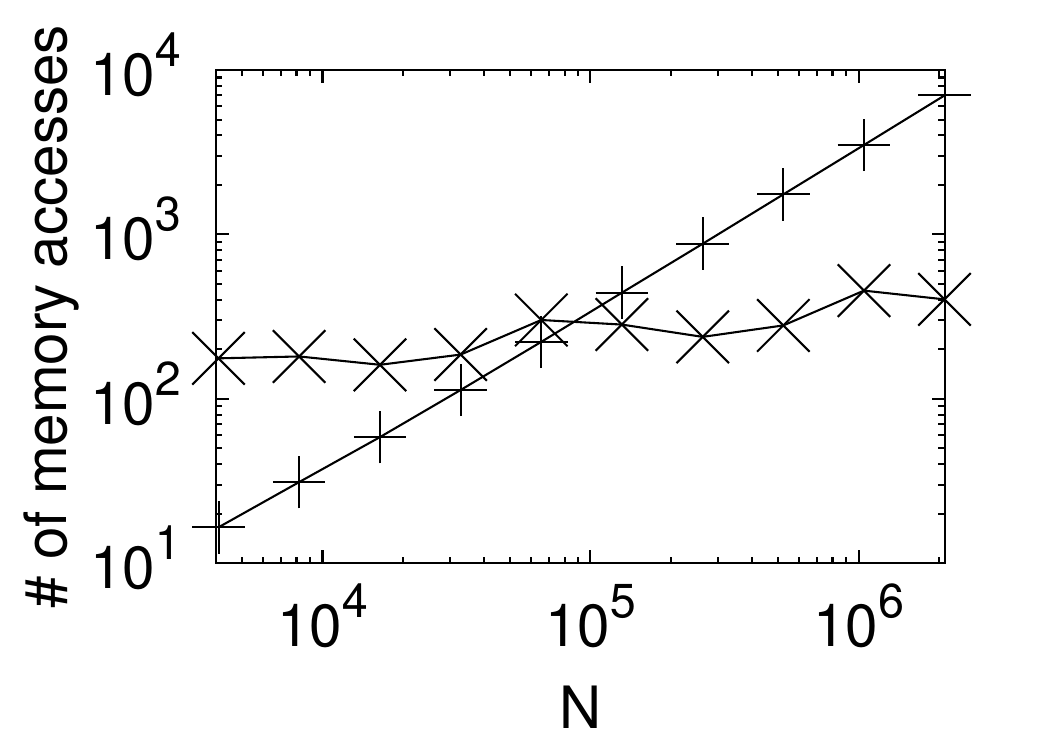}
    \end{minipage}
    &
    \begin{minipage}[htbp]{0.21\linewidth}
      \hspace{-2.3cm}
      \includegraphics[width=\linewidth]{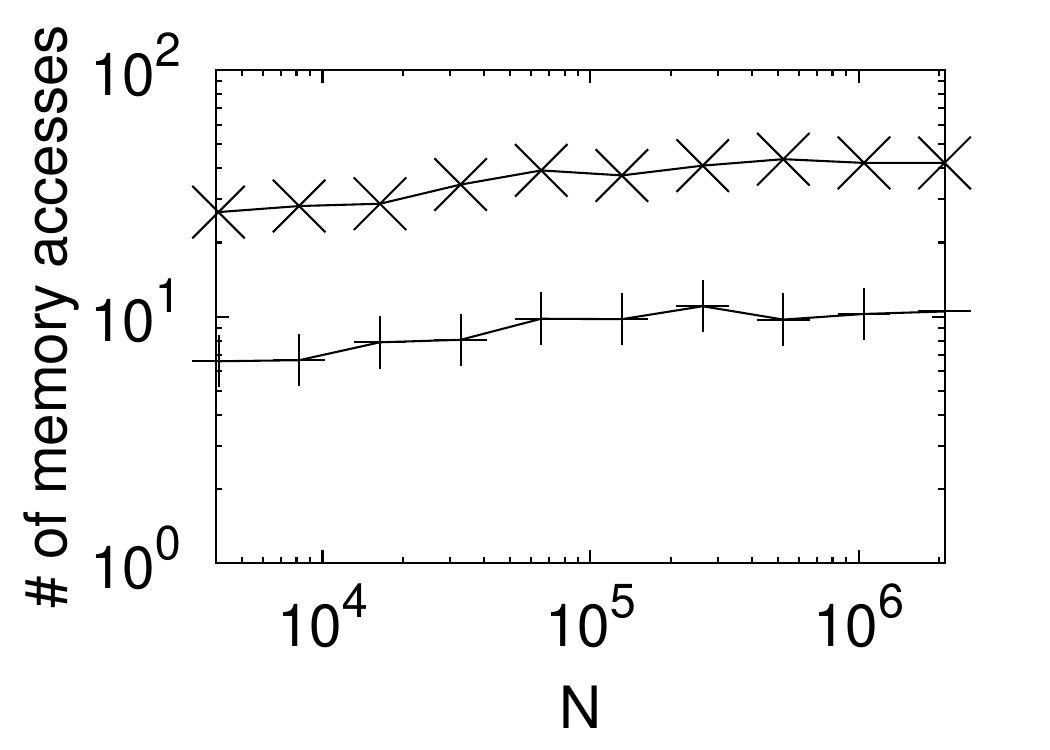}
    \end{minipage} \\
    \hspace{-0.35cm} (a) \emph{Brightkite} & \hspace{-1.2cm} (b) \emph{Gowalla} & \hspace{-2.2cm} (c) \emph{Brightkite} & \hspace{-3.1cm} (d) \emph{Brightkite} & \hspace{-3.8cm} (e) \emph{Brightkite}
  \end{tabular}
  \caption{The Effect of $N$ on (a) \& (b) Quantum Range Queries, (c) Multi-dimensional Range Queries (d) Dynamic Range Queries and (e) Insertions and Deletions into the Dynamic Data Structures} 
  \label{fig:exp:effect_n}
  \vspace*{-0.5cm}
\end{figure*}

\begin{figure}[tbp]
\vspace{-0.2cm}
  \centering
  \begin{minipage}[htbp]{\linewidth}
    \includegraphics[width=\linewidth]{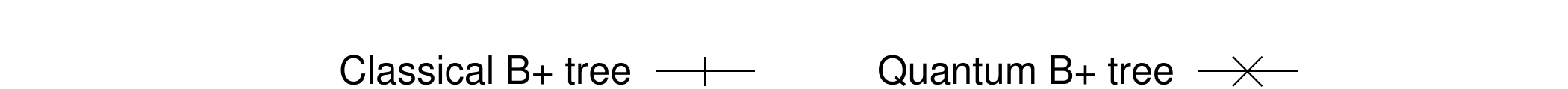}
    \vspace{-0.3cm}
   \end{minipage}
  \begin{tabular}{c c}
    \begin{minipage}[htbp]{0.44\linewidth}
      \includegraphics[width=\linewidth]{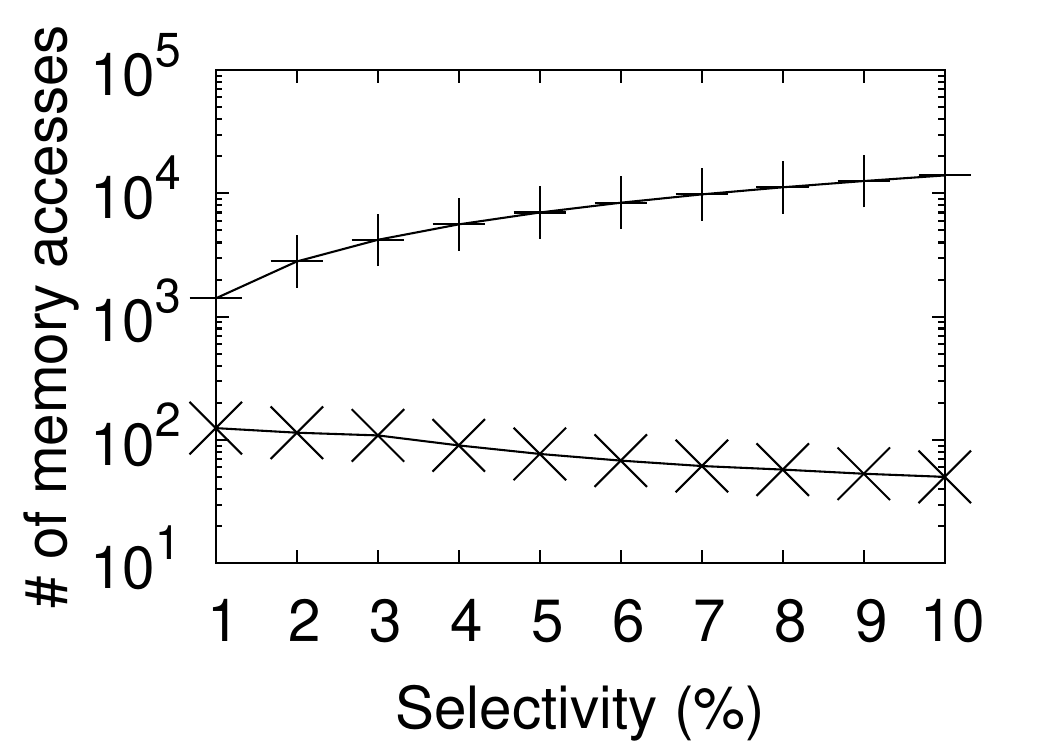}
    \end{minipage}
    &
    \begin{minipage}[htpb]{0.44\linewidth}
      \includegraphics[width=\linewidth]{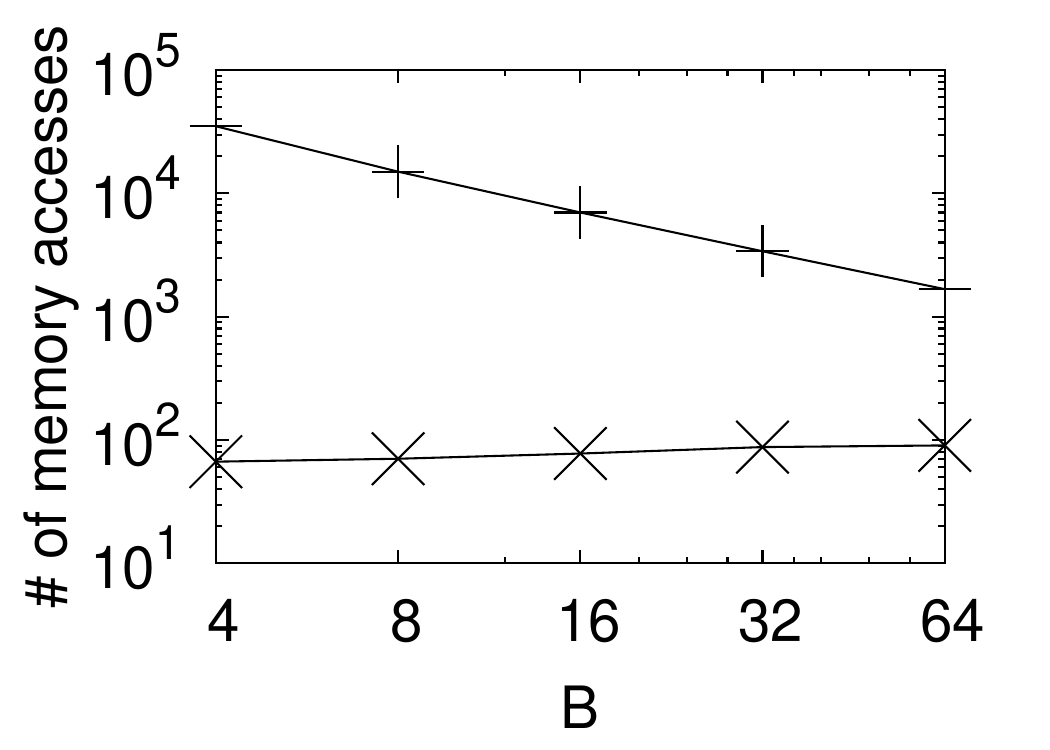}
    \end{minipage} \\
    (a) \emph{Brightkite} &
    (b) \emph{Brightkite} \\
    \begin{minipage}[htbp]{0.44\linewidth}
      \includegraphics[width=\linewidth]{figure/Brightkite-1-s.pdf}
    \end{minipage}
    &
    \begin{minipage}[htpb]{0.44\linewidth}
      \includegraphics[width=\linewidth]{figure/Brightkite-1-B.pdf}
    \end{minipage} \\
    (c) \emph{Brightkite} &
    (d) \emph{Brightkite}
  \end{tabular}
  \caption{Effect of Selectivity and $B$ of (a) \& (b) Quantum Range Queries and (c) \& (d) Dynamic Quantum Range Queries} 
  \label{fig:effect_s_b_static_dynamic}
\end{figure}

\smallskip
\noindent
\textbf{Effect of $N$.}
We first study the effect of the dataset $N$, which verifies the scalability of our proposed data structures and algorithms.
As shown in Figures~\ref{fig:exp:effect_n}(a) and (b),
our proposed \textbf{quantum B+ tree} with our proposed GCLQ search algorithm
scales well when $N$ grows from 4K to 2M on both datasets,
which verifies the efficient $O(\log_B N)$ cost of performing the GCLQ search on our \textbf{quantum B+ tree}.
As $N$ increases, the number of items within the query range (i.e., $k$)
for each query also tend to increase (approximately linearly with $N$).
Thus, the number of memory access for the \textbf{classical B+ tree} demonstrates a linear growth
with $N$, which complies its $O(\log_B N + k)$ complexity.
In comparison, the performance of our \textbf{quantum B+ tree} does not depend on the number of range query results,
and thus it is up to 1000x more efficient than the \textbf{classical B+ tree} for the (static) quantum range queries.

For multi-dimensional quantum range queries, our proposed \textbf{quantum range tree} obtains
superior performance than the \textbf{classical range tree} similarly,
as shown in Figure~\ref{fig:exp:effect_n}(c).
For the dynamic quantum range queries, although our (dynamic) \textbf{quantum B+ tree}
has larger number of memory accesses than the (dynamic) \textbf{classical B+ tree} for small $N$ (i.e., $<$ 100K),
our \textbf{quantum B+ tree} shows much better scalability for larger $N$,
as illustrated in Figure~\ref{fig:exp:effect_n}(d).
For the dynamic versions, we also tested the performance of insertions and deletions.
Following \cite{alsubaiee2014storage}, we insert all the records in the datasets
where for each insertion, there is 1\% chance to delete an existing record instead of doing this insertion,
and we measure the number of memory accesses for each insertion or deletion operation on average.
Figure~\ref{fig:exp:effect_n}(e) shows that the dynamic \textbf{quantum B+ tree}
needs 3x more memory accesses for insertion or deletion,
because the \textbf{quantum B+ tree} has more complex structure than the \textbf{classical B+ tree}.
However, the \textbf{quantum B+ tree} shows the $O(N)$ growth, which is similar to the \textbf{classical B+ tree},
indicating that the insertion and deletions operations on our \textbf{quantum B+ tree} are still reasonably efficient.

\smallskip
\noindent
\textbf{Effect of Selectivity.}
Then, we study the effect of selectivity on the quantum range queries.
When the selectivity (i.e., $\frac{k}{N}$) increases from 1\% to 10\% for both the static and the dynamic quantum range queries
(as shown in Figures~\ref{fig:effect_s_b_static_dynamic}(a) and (c)),
the number of memory accesses increases linearly for the \textbf{classical B+ tree}
due to the $O(\log_B N + k)$ query complexity.
Our proposed \textbf{quantum B+ tree} with $O(\log_B N)$ cost surprisingly
shows better performance as the selectivity increases.
This is because a larger $k$ shortened the process of the classical global search in our GCLQ search algorithm,
such that the efficient local quantum search can be triggered earlier.
Also, when $k$ increases, the post-selection will be accelerated
since the cost of post-selection is linear to $\frac{N}{k}$ as we mentioned in Section~\ref{subsec:static_range_query_alg}.

In the multi-dimensional quantum range queries,
as shown in Figure~\ref{fig:effect_s_b_md}(a),
although the number of memory accesses of our \textbf{quantum range tree} does not decrease as the selectivity increases
(because a larger $k$ leads to more candidate nodes from global classical search in the multi-dimensional case),
our \textbf{quantum range tree} is still efficient (since its performance does not depend on the selectivity) and is 10x--100x superior than the \textbf{classical range tree}.

\begin{figure}[tbp]
\vspace{-0.2cm}
  \centering
  \begin{minipage}[htbp]{\linewidth}
    \includegraphics[width=\linewidth]{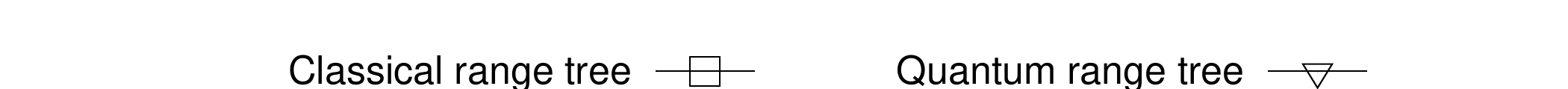}
    \vspace{-0.3cm}
   \end{minipage}
  \begin{tabular}{c c}
    \begin{minipage}[htbp]{0.44\linewidth}
      \includegraphics[width=\linewidth]{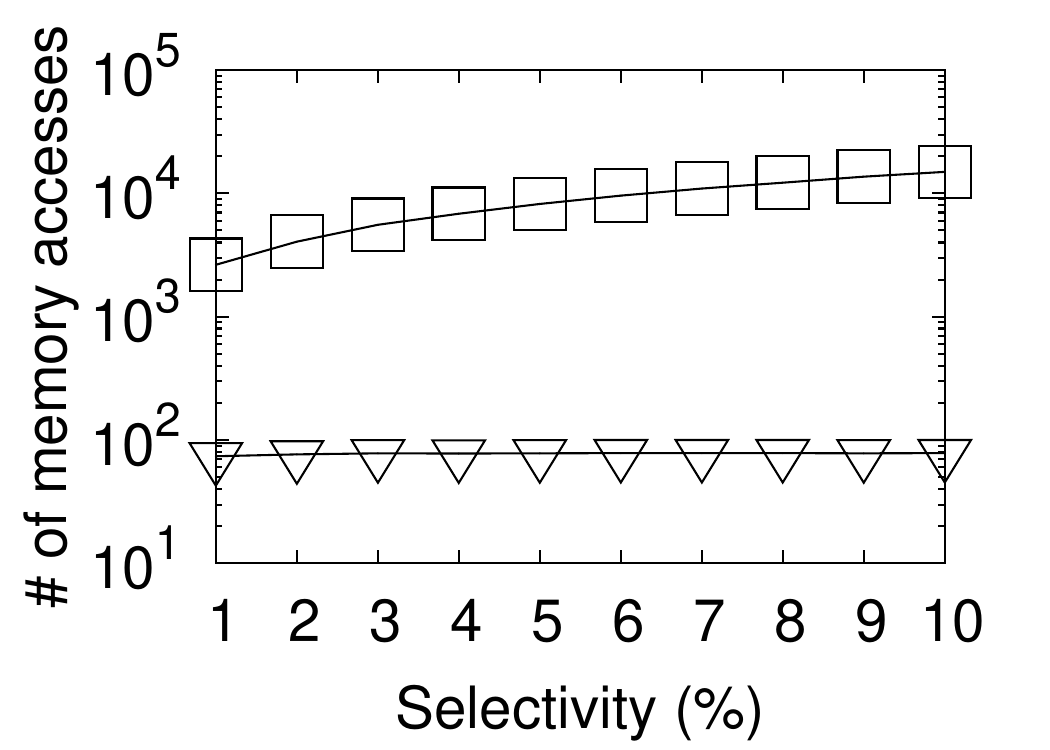}
    \end{minipage}
    &
    \begin{minipage}[htpb]{0.44\linewidth}
      \includegraphics[width=\linewidth]{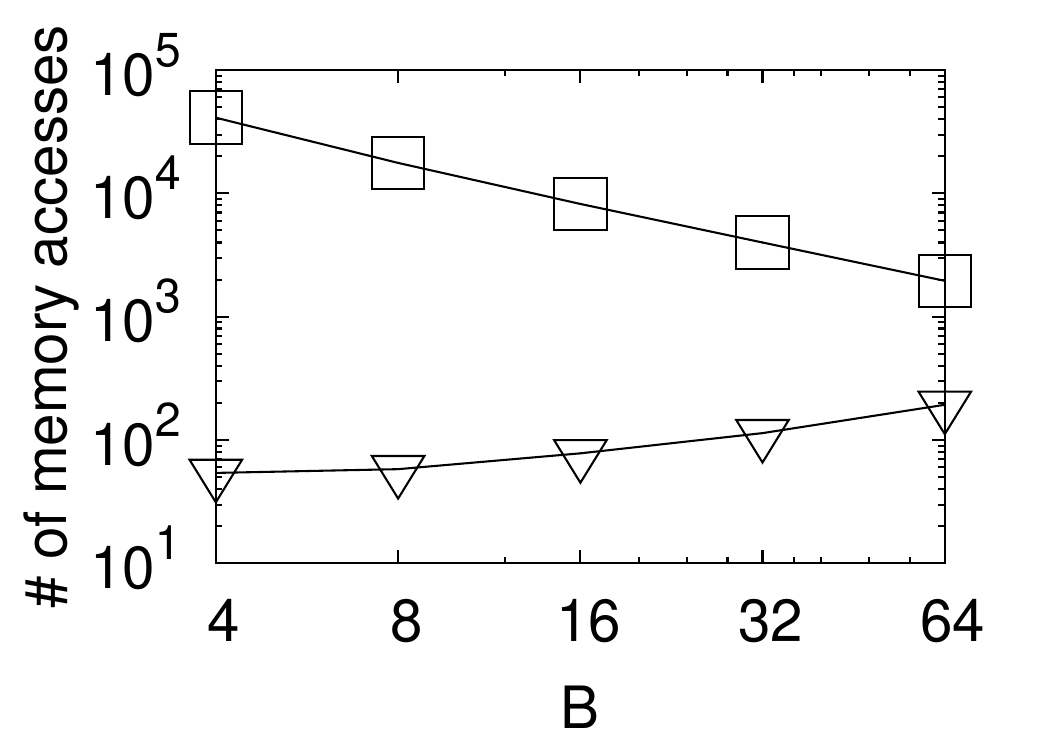}
    \end{minipage} \\
    (a) \emph{Brightkite} &
    (b) \emph{Brightkite}
  \end{tabular}
  \caption{Effect of Selectivity and $B$ of Multi-dimensional Quantum Range Queries} 
  \label{fig:effect_s_b_md}
\end{figure}

\smallskip
\noindent
\textbf{Effect of $B$.}
We also study the effect of the branching factor $B$.
As shown in Figures~\ref{fig:effect_s_b_static_dynamic}(b),~\ref{fig:effect_s_b_static_dynamic}(d)
and~\ref{fig:effect_s_b_md}(b),
when $B$ increases, the number of memory accesses for the classical data structures decreases sharply,
since the height of the tree is smaller for larger $B$.
Our quantum data structures needs slightly more memory accesses as $B$ increases,
because the success rate of the post-selection could be affected
as we mentioned in Section~\ref{subsubsec:local}.
However, it is well known that a larger branching factor leads to
more memory consumption for a classical B+ tree~\cite{comer1979ubiquitous} (and thus a quantum B+ as well).
Thus, we set $B$ to be 16 in other experiments for fair comparisons with a reasonable memory consumption.
Overall, our proposed quantum data structures favor a smaller branching factor.

\smallskip
\noindent
\textbf{Summary.}
In summary, our \textbf{quantum B+ tree} with our proposed GCLQ search algorithm
achieves up to 1000x performance improvement than the \textbf{classical B+ tree}.
On the dataset \emph{Brightkite} of size 2M,
the average number of memory access is only around 40 for the \textbf{quantum B+ tree}
for the quantum range query,
while the \textbf{classical B+ tree} needs around 40K memory accesses.
The similar superiority of our \textbf{quantum B+ tree}
is observed on the dynamic and multi-dimensional quantum range queries
compared with the classical data structures.
We also show that our \textbf{quantum B+ tree} scales well with the dataset size,
the selectivity of the query ranges and the branching factor.

\section{Related Work}\label{sec:related}

The classical range query problem has been studied for decades with various data structures proposed to solve this problem.
There is little room for further significant improvement.
For the static range query and the dynamic range query,
the B-tree \cite{bayer2002organization} and the B+ tree \cite{comer1979ubiquitous} are widely-used.
The B+ tree is a tree data structure which can find a key in $O(\log_B N)$ time,
where $N$ is the number of key-record pairs and $B$ is the branching factor.
Then, $O(k)$ time is needed to load the $k$ key-record pairs in the query range,
so the range query on a B+ tree costs $O(\log_B N + k)$ time,
which is shown to be asymptotically optimal in classical computers~\cite{yao1981should}.

For the high-dimensional static range query, Bentley \cite{bentley1978decomposable} proposed the range tree to answer it in $O(\log^d N + k)$ time, where $d$ is the dimension of the keys.
In addition, the range tree needs $O(N\log^{d-1} N)$ storage space.
Chazelle later \cite{chazelle1990lower1, chazelle1990lower2} proposed the lower bounds for the $d$-dimensional cases: $O(k+polylog(N))$ time complexity with $\Omega(N(\log N/\log$ $\log N)^{d-1}))$ storage space and $\Omega((\log N/$ $\log (2C/N))^{d-1}+k)$ time complexity with $C$ units storage space where the lower bound of query I/O cost is provably tight for $C=\Omega(N(\log N)^{d-1})$. 

However, all the above classical data structures have the same problem that the execution time grows linearly with $k$, which makes them useless for quantum algorithms.

The quantum database searching problem is also very popular in the quantum algorithm field. Grover's algorithm is described as a database search algorithm in \cite{grover1996fast}. It solves the problem of searching a marked record in an unstructured list, which means that all the $N$ records are arranged in random order. On average, the classical algorithm needs to perform $N/2$ queries to a function $f$ which tells us if the record is marked. More formally, for each index $i$, $f(i)=1$ means the $i$-th record is marked and $f(i)=0$ means the $i$-th record is unmarked. Note that only one record is marked in the database.
Taking the advantage of quantum parallelism, Grover's algorithm can find the index of the marked record with $O(\sqrt{N})$ queries to the oracle. The main idea is to first ``flip" the amplitude of the answer state and then reduce the amplitudes of the other states. One such iteration will enlarge the amplitude of the answer state and $O(\sqrt{N})$ iterations should be performed until the probability that the qubits are measured to be the answer comes close to 1. Then, an improved Grover's algorithm was proposed in \cite{boyer1998tight}. We are also given the function $f$ to mark the record, but $k$ records are marked at this time. The improved Grover's algorithm can find one of the $k$ marked records in $O(\sqrt{N/k})$ time. If we make the function $f$ to mark the records with keys within a query range, then this algorithm returns one of the $k$ key-record pairs in $O(\sqrt{N/k})$ time, and it needs $O(\sqrt{Nk})$ time to answer a range query. 
Since the query time grows linearly with the square root of $N$,
this algorithm is much less efficient when $N$ is very large,
even compared to using the classical data structure with $O(\log N + k)$ time complexity.
The reason is that this algorithm only handle the unstructured dataset and does not leverage the power of data structures that could be pre-built as a database index.

In the database area, there are also plenty of studies discussing how to use quantum computers to further improve traditional database queries. For example, \cite{uotila2022synergy, schonberger2022applicability, fankhauser2021multiple, trummer2015multiple} discussed quantum query optimization, which is to use quantum algorithms like quantum annealing \cite{finnila1994quantum} to optimize a traditional database query. However, compared with the quantum range query discussed in this paper, the existing studies are in a different direction. The existing studies are discussing how to use a quantum algorithm to optimize range queries in a classical computer, where the query returns a list of records. In this paper, we discuss how to use a classical algorithm to optimize range queries in a quantum computer, where the query returns quantum bits in a superposition of records. We focus on a quantum data structure stored in a quantum computer, where the classical algorithm is only to assist the query. To our best knowledge, in the database area, we are the first to discuss quantum algorithms in this direction.

In conclusion, the existing classical data structures and the existing quantum database searching algorithms cannot solve the range query problem in quantum computers perfectly.

\section{Conclusion}\label{sec:con}

In this paper, we study the quantum range query problem.
We propose the quantum B+ tree, the first tree-like quantum data structure,
and the efficient global-classical local-quantum search algorithm based on the quantum B+ tree.
Our proposed data structure and algorithm can answer a quantum range query in $O(\log_B N)$ time,
which is asymptotically optimal in quantum computers
and is exponentially faster than classical B+ trees.
Furthermore, we extend it to a dynamic quantum B+ tree.
The dynamic quantum B+ tree can support insertions and deletions in $O(\log_B N)$ time
and answer a dynamic quantum range query in $O(\log^2_B N)$ time.
We also extend the quantum B+ tree to the quantum range tree
to solve the $d$-dimensional quantum range query in $O(\log^d_B N)$ time.
In the experiments, we did simulations to verify the superiority of our proposed quantum data structures compared with the classical data structures.
We expect that the quantum data structures will show significant advantages in the real world.

The future work includes exploring even more efficient quantum algorithms for the dynamic and multi-dimensional range queries
and studying more advanced database queries such as the top-$k$ queries.

\bibliographystyle{ACM-Reference-Format}
\bibliography{reference}

\appendix

\begin{figure*}[tbp]
  \centering
  \hspace*{-1.5cm}
  \includegraphics[width=0.95\linewidth]{figure/legend3.pdf}
  
  \begin{tabular}{c c c c}
    \begin{minipage}[htbp]{0.255\linewidth}
      \hspace{-0.4cm}
      \includegraphics[width=\linewidth]{figure/Gowalla-1.pdf}
    \end{minipage}
    &
    \begin{minipage}[htbp]{0.255\linewidth}
      \hspace{-0.8cm}
      \includegraphics[width=\linewidth]{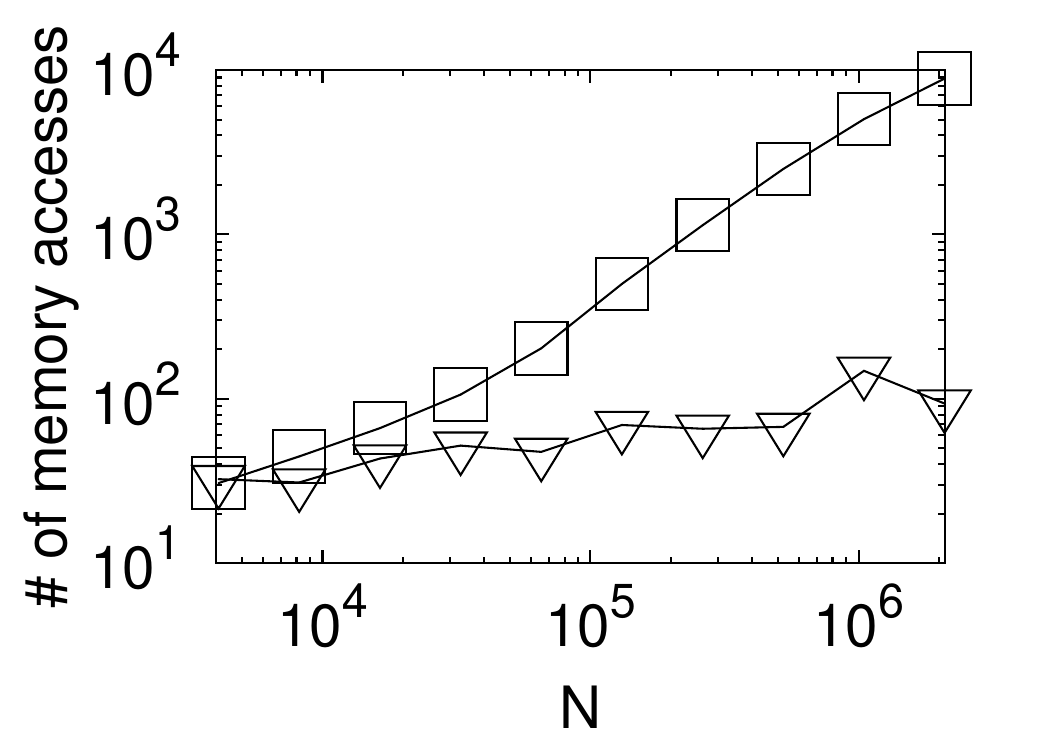}
    \end{minipage}
    &
    \begin{minipage}[htbp]{0.255\linewidth}
      \hspace{-1.2cm}
      \includegraphics[width=\linewidth]{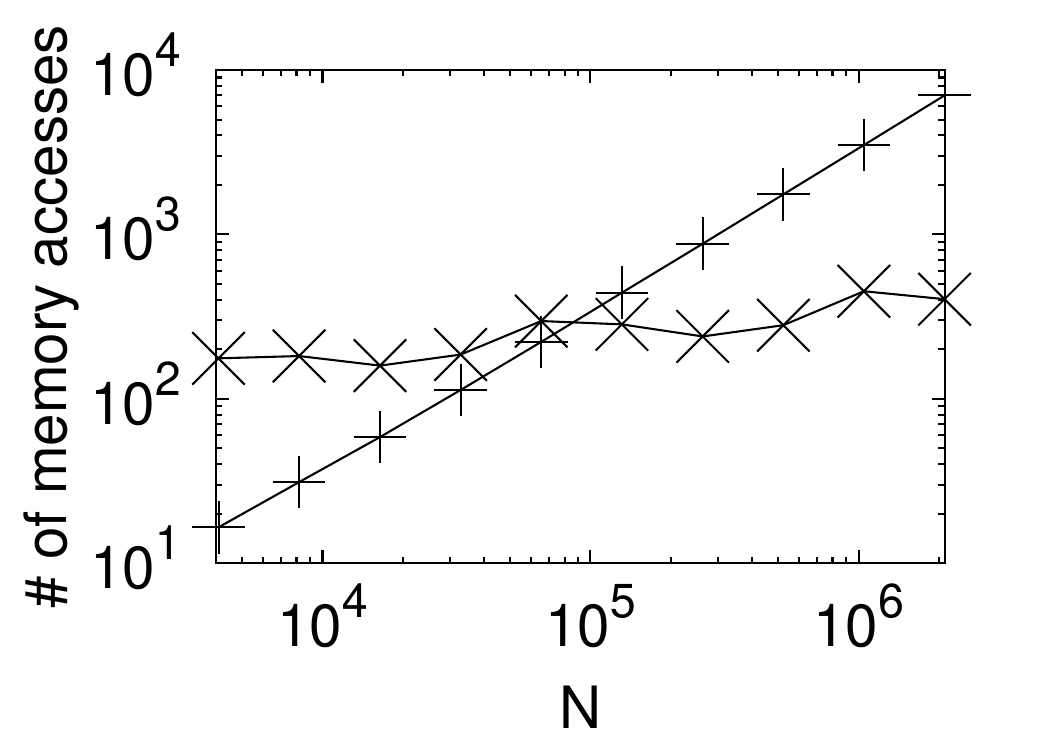}
    \end{minipage}
    &
    \begin{minipage}[htbp]{0.255\linewidth}
      \hspace{-1.5cm}
      \includegraphics[width=\linewidth]{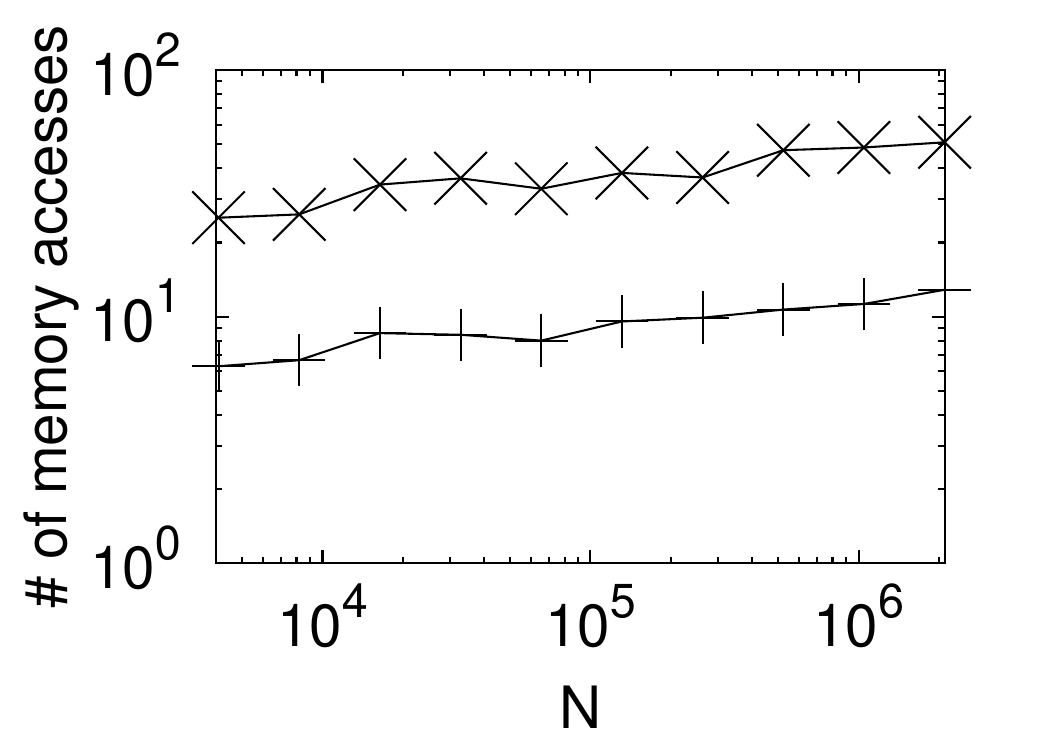}
    \end{minipage} \\
    \hspace{-0.35cm} (a) \emph{Gowalla} & \hspace{-1.2cm} (b) \emph{Gowalla} & \hspace{-2.2cm} (c) \emph{Gowalla} & \hspace{-3.1cm} (d) \emph{Gowalla}
  \end{tabular}
  \caption{The Effect of $N$ on (a) Quantum Range Queries, (b) Multi-dimensional Range Queries (c) Dynamic Range Queries and (d) Insertions and Deletions into the Dynamic Data Structures} 
  \label{fig:exp:effect_n_g}
  \vspace*{-0.5cm}
\end{figure*}

\begin{figure}[tbp]
\vspace{-0.2cm}
  \centering
  \begin{minipage}[htbp]{\linewidth}
    \includegraphics[width=\linewidth]{figure/legend1.pdf}
    \vspace{-0.3cm}
   \end{minipage}
  \begin{tabular}{c c}
    \begin{minipage}[htbp]{0.44\linewidth}
      \includegraphics[width=\linewidth]{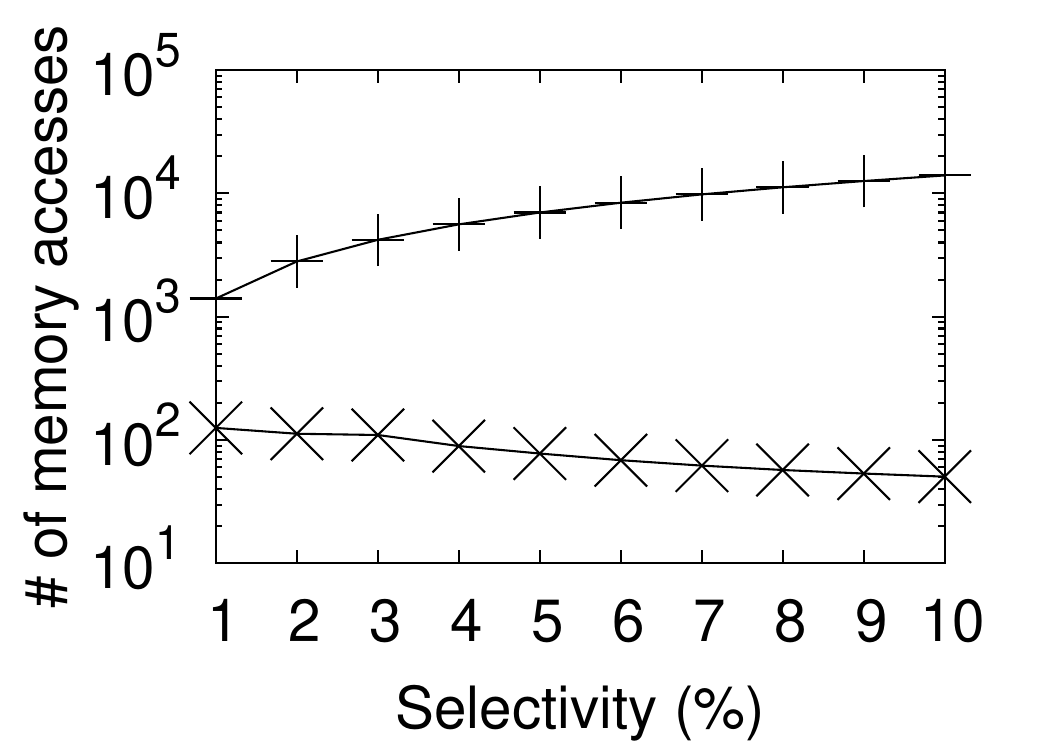}
    \end{minipage}
    &
    \begin{minipage}[htpb]{0.44\linewidth}
      \includegraphics[width=\linewidth]{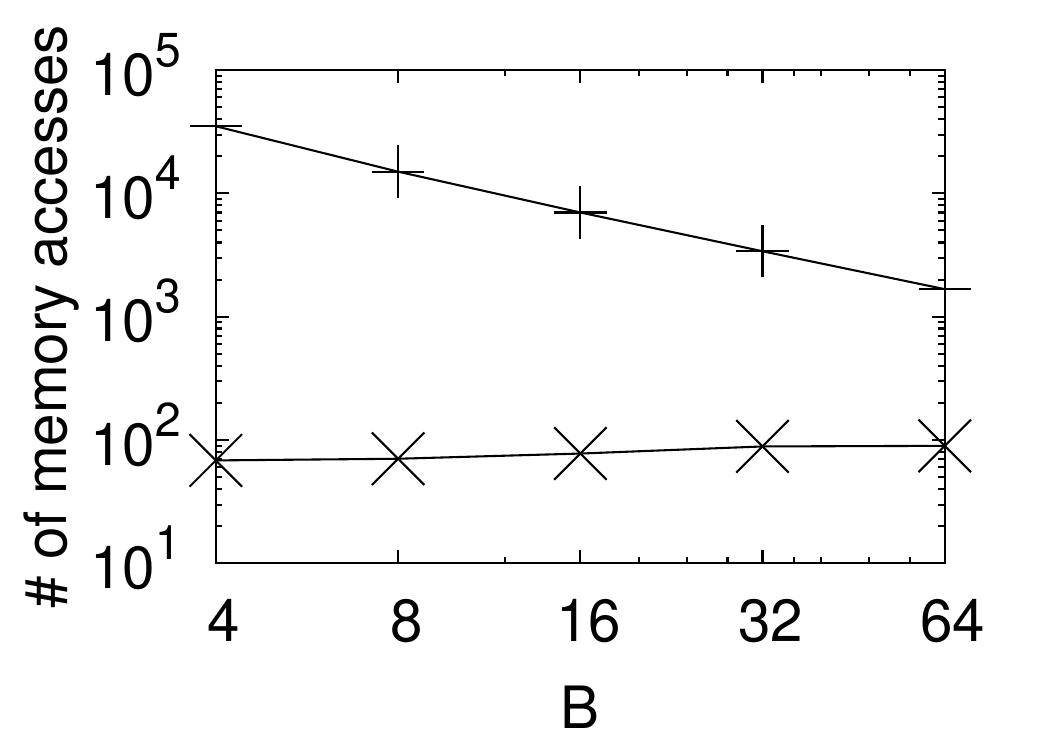}
    \end{minipage} \\
    (a) \emph{Gowalla} &
    (b) \emph{Gowalla} \\
    \begin{minipage}[htbp]{0.44\linewidth}
      \includegraphics[width=\linewidth]{figure/Gowalla-1-s.pdf}
    \end{minipage}
    &
    \begin{minipage}[htpb]{0.44\linewidth}
      \includegraphics[width=\linewidth]{figure/Gowalla-1-B.pdf}
    \end{minipage} \\
    (c) \emph{Gowalla} &
    (d) \emph{Gowalla}
  \end{tabular}
  \caption{Effect of Selectivity and $B$ of (a) \& (b) Quantum Range Queries and (c) \& (d) Dynamic Quantum Range Queries} 
  \label{fig:effect_s_b_static_dynamic_g}
\end{figure}

\begin{figure}[tbp]
\vspace{-0.2cm}
  \centering
  \begin{minipage}[htbp]{\linewidth}
    \includegraphics[width=\linewidth]{figure/legend2.pdf}
    \vspace{-0.3cm}
   \end{minipage}
  \begin{tabular}{c c}
    \begin{minipage}[htbp]{0.44\linewidth}
      \includegraphics[width=\linewidth]{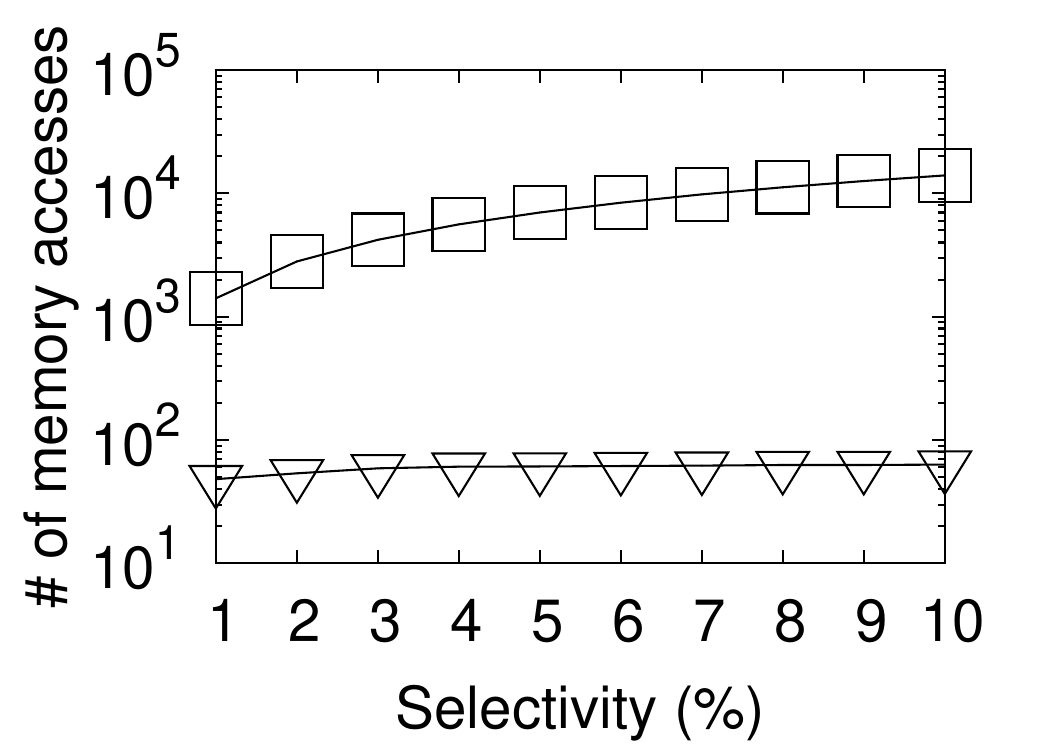}
    \end{minipage}
    &
    \begin{minipage}[htpb]{0.44\linewidth}
      \includegraphics[width=\linewidth]{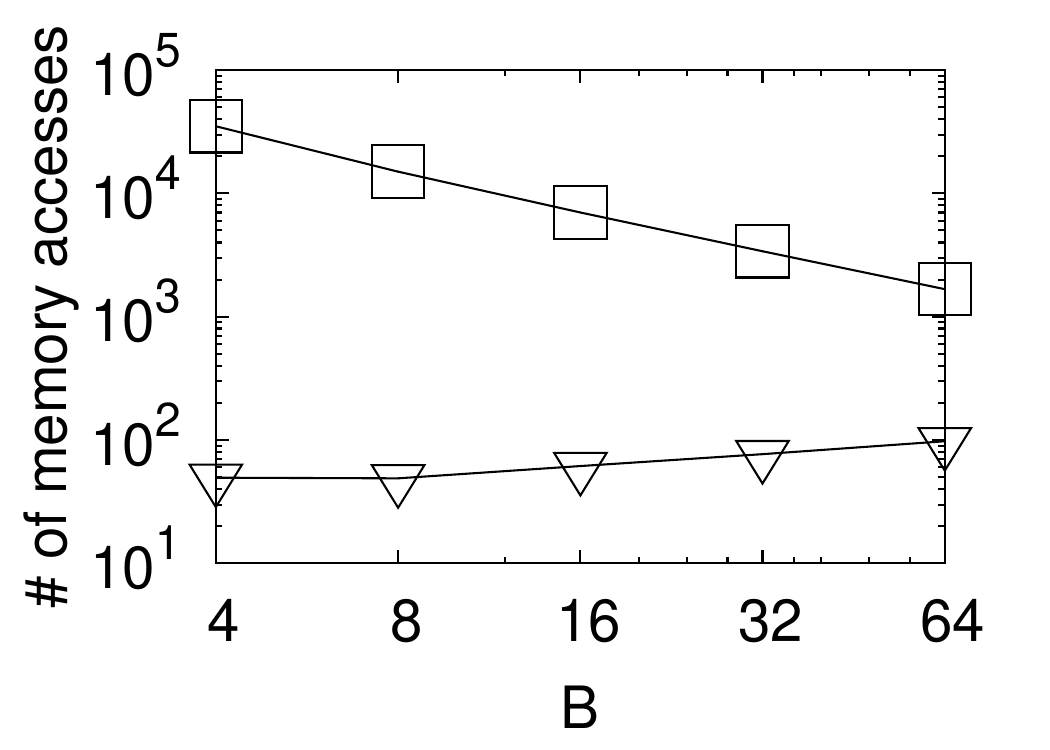}
    \end{minipage} \\
    (a) \emph{Gowalla} &
    (b) \emph{Gowalla}
  \end{tabular}
  \caption{Effect of Selectivity and $B$ of Multi-dimensional Quantum Range Queries} 
  \label{fig:effect_s_b_md_g}
\end{figure}

\section{Remaining Experimental Results}
In this section, we show the complete experimental results of dataset \emph{Gowalla},
which are not presented in Section~\ref{exp} due to the space limit.

\section{Complete Proofs}
In this section, we show the complete proofs of all the lemmas and theorems in this paper.

\begin{proof}[\textbf{Proof of Lemma~\ref{lem:lower_bound_time_complexity}}]
Consider a static quantum range query $QUERY(x, x)$.
It answers whether the key $x$ exists in the dataset,
which is a static membership problem.
By \cite{sen2001lower}, the time complexity of answering a static membership problem
in quantum computers is $\Omega(\log N)$.
Therefore, the time complexity of answering a static quantum range query is also $\Omega(\log N)$.
\end{proof}

\begin{proof}[\textbf{Proof of Lemma~\ref{lem:global_classical}}]
    Since the routing keys of all nodes in the same level are disjoint,
    we cannot have more than two partial nodes in the same level.
    It is also easy to verify that the returned nodes are from the same level,
    Thus, the candidate set contains at most two nodes.
    Since we only filter out the outside node in this algorithm,
    we have $\mathcal{R}^* \subset \mathcal{R}$.
    Finally, since the precise partial node either is a leaf
    (which contains at least $1/B$ items in the query range),
    or contains an inside child
    (which contains at least $\frac{1}{4B}$ items in the query range due to the balanced nodes of the B+ tree),
    and it can be verify that at least one returned node is precise,
    we have $\vert \mathcal{R}^* \vert \geq \frac{1}{8B} \vert \mathcal{R} \vert$.
\end{proof}

\begin{proof}[\textbf{Proof of Lemma~\ref{the:time_complexity_gclq}}]
The global classical search costs $O(\log_B N)$ time because we return at most two candidates and the height of the tree is $O(\log_B N)$.
For the local quantum search, we repeat all the steps for $\frac{2B^{h-j+1}}{k}$ times on average.
By the condition to trigger a local quantum search, 
all the non-dummy key-record pairs below one of the children of $u$ and $v$ are all in the answer, therefore $k\geq \frac{1}{4}B^{h-j}$ by the definition of our quantum B+ tree. So, we need to repeat all the steps for at most $8B$ times, which is a constant time. In each iteration, we do Step 2 and Step 3 for at most $O(\log_B N)$ times, so the local quantum search needs $O(\log_B N)$ time.
Therefore, the quantum range query algorithm needs $O(\log_B N)$ time.
\end{proof}

\begin{proof}[\textbf{Proof of Theorem~\ref{lem:insertion_cost}}]
By Theorem \ref{theo1} in this paper and Theorem 3.1 in \cite{bentley1980decomposable}, $N$ insertions totally cost $O(N\log_B N)$ time. Therefore, the amortized cost of one insertion is $O(\log_B N)$.

To analyze the update cost in a merge, we have that when merging $F_i$, the time complexity to update $T_1$ is $O(\log_B N)$.
The reason is as follows.

Let $ID_l$ denote the least ID in $F_i$. For each $j<i$ and each key-record pair in $F_j$, the ID of the key-record pair is smaller than $ID_l$. Let $ID_r$ denote the greatest ID in $F_i$. Then, for each $j>i$ and each key-record pair in $F_j$, the ID of the key-record pair is greater than $ID_r$. Therefore, the update operations for the key-record pairs in $F_i$ can be merged into a range update. Then, we can use lazy propagation \cite{ibtehaz2021multidimensional} to do the range update in $O(\log_B N)$ time.

Thus, the extra cost to maintain $T_0$ and $T_1$ has no impact, 
which means the amortized cost of insertion is still $O(\log_B N)$.

For deletions, 
the first step costs $O(\log_B N)$ time, since it consists of two point queries in B+ trees.

Then consider the second step. Motivated by the analysis of partial rebuilding in \cite{overmars1983design}, we consider a node $u$ of height $h$ just after a rebuild. The node $u$ is perfectly balanced such that its weight $w(u)\geq \frac{1}{2}B^{h+1}$. Since it will become imbalanced if and only if $w(u)<\frac{1}{4}B^{h+1}$, there must be $\Omega(B^{h+1})$ deletions below the node $u$ or its siblings before that. So, it is charged $O(1)$ time for each deletion below it and its siblings. Then, for a deletion in a leaf node, each ancestor and its siblings are charged $O(1)$ time, so they are totally charged $O(\log_B N)$ time.

Therefore, the amortized cost of a deletion is $O(\log_B N)$.

\end{proof}

\begin{proof}[\textbf{Proof of Theorem~\ref{the:time_complexity_dynamic}}]
As above mentioned, the global classical search costs $O(\log^2_B N)$ time. Consider the cost of the local quantum search. The initialization and the steps to obtain the leaves cost $O(\log_B N)$. Then, we multiply it by the average post-selection times. Let $N'$ denote the total number of key-record pairs below the nodes in $L_0,\cdots, L_i$. Then, 
the post-selection needs to be iterated $O(N'/k)$ times on average. Since there are at most $2B\lfloor\log_B N\rfloor$ nodes in the lists and each node has a weight at most $B^{i+1}$, we have $N'\leq 2B^{i+2}\log_B N$. Before we turn to the local quantum search, we have found a child of a node in $L_i$ such that all the key-record pairs below the child are in the answer, such that $k\geq \frac{1}{4}B^i$. Therefore, $N'/k=O(\log_B N)$. Hence, the local quantum search costs $O(\log^2_B N)$ time on average.

Since both the global classical search and the local quantum search cost $O(\log_B^2 N)$ time, the dynamic quantum B+ tree answers a range query in $O(\log^2_B N)$ time on average.
\end{proof}

\begin{proof}[\textbf{Proof of Theorem~\ref{the:time_complexity_md}}]
The first step is to recursively search the quantum range trees to obtain the $O(\log_B^{d-1} N)$ $1$-dimensional quantum range trees which contain answers. It costs $O(\log_B^{d-1} N)$ time. The second step is to perform a global classical search on the $O(\log_B^{d-1} N)$ quantum B+ tree. It costs $O(\log_B^d N)$ time. The third step is to perform a local quantum search starting from $O(\log_B^{d-1} N)$ candidate nodes returned in the second step. We first initialize the quantum bits in $O(\log_B^{d-1} N)$ time. Then, we do the same Step 2 and Step 3 in Section \ref{local} to obtain the leaves below the candidate nodes in $O(\log_B N)$ time.
Thus, we need to do a constant number of post-selections. Therefore, this step costs $O(\log_B^{d-1} N)$ time.

Therefore, the quantum range tree answers a range query in $O(\log^d_B N)$ time on average.
\vspace{-1mm}
\end{proof}

\end{document}